\documentclass{article}

\usepackage{microtype}
\usepackage{graphicx}
\usepackage{booktabs} 

\usepackage{hyperref}

\usepackage[accepted]{icml2025}

\makeatletter
\renewcommand{\Notice@String}{Preprint. Under review.}
\makeatother

\usepackage{amsmath}
\usepackage{amsfonts}
\usepackage{amssymb}
\usepackage{mathtools}
\usepackage{amsthm}
\usepackage{enumitem}
\usepackage{subcaption}
\usepackage{tikz}
\usetikzlibrary{arrows.meta,decorations.pathreplacing}

\usepackage{comment}

\usepackage{aliascnt}
\usepackage[capitalize,noabbrev]{cleveref}
\usepackage{xspace}

\makeatletter
\newtheoremstyle{icml-note-it}
  {3pt}{3pt}
  {\normalfont}
  {}
  {\bfseries}
  {.}
  { }
  {\thmname{#1}\thmnumber{ #2}\thmnote{ \normalfont\itshape (#3)}}
\makeatother

\theoremstyle{icml-note-it}

\newaliascnt{proposition}{theorem}
\newtheorem{proposition}[proposition]{Proposition}
\aliascntresetthe{proposition}

\newaliascnt{lemma}{theorem}
\newtheorem{lemma}[lemma]{Lemma}
\aliascntresetthe{lemma}

\newaliascnt{corollary}{theorem}
\newtheorem{corollary}[corollary]{Corollary}
\aliascntresetthe{corollary}

\newtheorem{assumption}{Assumption}

\newaliascnt{remark}{theorem}
\newtheorem{remark}[remark]{Remark}
\aliascntresetthe{remark}

\crefname{assumption}{assumption}{assumptions}
\Crefname{assumption}{Assumption}{Assumptions}
\newcounter{proof}

\crefname{proof}{proof}{proofs}
\Crefname{proof}{Proof}{Proofs}

\crefname{appsec}{appendix}{appendices}
\Crefname{appsec}{Appendix}{Appendices}

\crefname{remark}{remark}{remarks}
\Crefname{remark}{Remark}{Remarks}

\makeatletter
\renewenvironment{proof}[1][\proofname]{%
  \par\pushQED{\qed}%
  \normalfont\topsep6\p@\@plus6\p@\relax
  \trivlist
  \item[\hskip\labelsep\bfseries\upshape #1\@addpunct{.}]\ignorespaces
}{%
  \popQED\endtrivlist\@endpefalse
}
\makeatother

\newenvironment{refproof}[1][\proofname]{%
  \refstepcounter{proof}%
  \begin{proof}[#1]%
}{%
  \end{proof}%
}

\newcommand{\sw}{\ensuremath{\mathrm{sw}}}
\newcommand{\user}{\ensuremath{\mathrm{user}}}

\newcommand{\sub}{\ensuremath{\mathrm{sub}}}

\newcommand{\AIa}{\ensuremath{\mathrm{RE}_{\text{ad}}}\xspace}
\newcommand{\AIf}{\ensuremath{\mathrm{RE}_{\text{free}}}\xspace}
\newcommand{\AIp}{\ensuremath{\mathrm{RE}_{\text{paid}}}\xspace}

\definecolor{midnightgreen}{rgb}{0.0, 0.29, 0.33}
\newcommand{\cx}[1]{\textcolor{midnightgreen}{\bf\small [#1 --cx]}}
\newcommand{\bl}[1]{\textcolor{red}{\bf\small [#1 --bl]}}
\newcommand{\lz}[1]{\textcolor{blue}{\bf\small [#1 --lz]}}
\newcommand{\xc}[1]{\textcolor{orange}{\bf\small [#1 --chen]}}

\definecolor{lavender}{RGB}{147, 112, 219}

\icmltitlerunning{An Economic Framework for Generative Engines: Advertising or Subscription?}

\begin{document}

\twocolumn[
\icmltitle{An Economic Framework for Generative Engines: Advertising or Subscription?}


\icmlsetsymbol{equal}{*}

\begin{icmlauthorlist}
\icmlauthor{Luyang Zhang}{yyy,comp}
\icmlauthor{Cathy Jiao}{yyy}
\icmlauthor{Beibei Li}{yyy}
\icmlauthor{Chenyan Xiong}{yyy}
\end{icmlauthorlist}

\icmlaffiliation{yyy}{Carnegie Mellon University}
\icmlaffiliation{comp}{Anaxi Labs}

\vskip 0.15in
\begin{center}
\small
\textsuperscript{1}Carnegie Mellon University \\
\textsuperscript{2}Anaxi Labs
\end{center}

\icmlkeywords{Machine Learning, ICML}

\vskip 0.3in
]

\begin{abstract}
Generative Engines (GEs) such as ChatGPT and Google's AI Overviews are rapidly reshaping search economics by delivering synthesized responses that allow users to bypass third-party websites, cutting those sites' advertising revenue. Yet this shift also leaves GEs facing their own monetization problem: whether to insert ads into synthesized responses or keep them ad-free to drive subscription conversions. In this paper, we introduce a dynamic framework to study this problem, which captures how query-level design choices shape user engagement, retention, and subscription conversion over time. Using this framework, we show that the optimal policy follows a cutoff rule: ads should only be shown to users only when the immediate ad payoff exceeds the long-term value of providing ad-free responses. This cutoff shifts toward \textit{with-ad responses} when i) ad revenue is high or ii) users are less sensitive to ads, and toward \textit{ad-free responses} when iii) subscription conversion becomes relatively more valuable. In addition, the presence of rival GEs shifts the optimal policy further toward ad-free responses, as ad-heavy monetization becomes less sustainable when users can freely switch to alternatives. Our findings reveal incentives for real-life generative engine providers to adopt designs that enhance user experience and long-term sustainability. 
\end{abstract}

\section{Introduction}

Generative engines (GEs), such as ChatGPT and Google's AI Overviews, use large language models to present information directly in response to user queries \citep{zhou2024understanding,scarcella2026google}, and are being rapidly adopted in web search \citep{Venkatachary2024AIOverviewsSearch, OpenAI2024IntroducingChatGPTSearch, era_ge_kaiser_2025}. Consequently, GEs have changed user behavior by satisfying information needs directly, reducing how often users visit external websites, and thereby challenging the ad-based monetization models that websites and search engines have traditionally relied on \citep{Gleason2023SERPComponents,pew2025_ai_summary_clicks,SeerInteractive2025AIOCTRUpdate} 

This paradigm shift presents a new problem for ad monetization in GEs. Current generative engine providers have already adopted diverse monetization approaches; for instance, OpenAI and Google have integrated advertising into their generative AI products \citep{googleads2024ai-max-search-campaigns, openai_chatgpt_go_2026}, while Anthropic remains strictly ad-free \citep{ChmielewskiSeetharamanCherney2026AnthropicSuperBowlAds}. Yet it remains unclear which approach is most viable in the long term. In particular, because GEs incur substantial model inference costs \citep{luccioni2023estimating}, providers face a trade-off between \textit{with-ad} and \textit{ad-free} subscription models, a choice with consequences for user retention, subscription growth, and monetization sustainability \citep{reuters2026openai}.

To address these issues, we develop a framework to analyze the economics of monetization in GEs. Specifically, we model the interaction between GEs and users as a dynamic Stackelberg game. At each query, the GE leads by choosing to provide either an with-ad or ad-free response, after which the user decides whether to engage -- and eventually become a subscribed customer -- or take an outside option -- such as a competing GE -- following a discrete-choice model \citep{mcfadden1972conditional}. The GE anticipates these responses and adopts a monetization policy accordingly. Crucially, these interactions accumulate over time: ad-free responses can strengthen future engagement and subscription demand \citep{cox1972regression,rust1987optimal}, while repeated ad exposure can erode retention. This dynamic structure allows us to analyze how optimal monetization design varies across user and query types, as well as across market conditions.

Next, within our framework, we characterize the generative engine's \textit{optimal design policy} and study its economic implications. We show that the optimal policy takes a context-dependent threshold form. Specifically, the GE provides with-ad response only when i) its immediate ad revenue exceeds the long-term value of retaining the user ad-free (i.e., it exceeds the expected future revenue from stronger user retention and eventual subscription conversion), and ii) when users are less ad-sensitive \citep{gupta2006modeling,ascarza2013joint}. Conversely, GEs shift toward ad-free responses when ad-free interactions generate more long-term value. Moreover, stronger external competition also shifts the threshold toward ad-free responses, since users can more easily switch to competing GEs (see \Cref{fig:intro_policy_overview}).

Finally, we conduct market simulations to examine how GE monetization policies, ranging from ad-heavy to ad-free, perform over time. We find that while ad-heavy policies generate higher short-term revenue, they erode the user base and slow subscription growth in the long run. Moreover, the costs of ad-heavy policies are not easily undone, as switching to ad-free responses later does not replicate the gains of having served them consistently from the start. These patterns hold across user types, query types, and market conditions, and intensify under stronger outside competition.

To summarize, our contributions are three-fold:
\begin{enumerate}[noitemsep, topsep=2pt, parsep=0pt, partopsep=0pt]
    \item 
    We provide a game-theoretic framework modeling the interaction between GEs and users around with-ad vs. ad-free responses
    
    \item 
    We characterize the optimal design policy in threshold form, highlighting trade-offs between short-term ad revenue and long-term user engagement.

    \item 
    We conduct market simulations comparing ad-heavy to ad-free policies over time, finding that ad-heavy strategies hurt long-term retention and subscription growth even if they boost short-term revenue
    
\end{enumerate}

\begin{figure}[t]
  \centering
  \includegraphics[width=\linewidth]{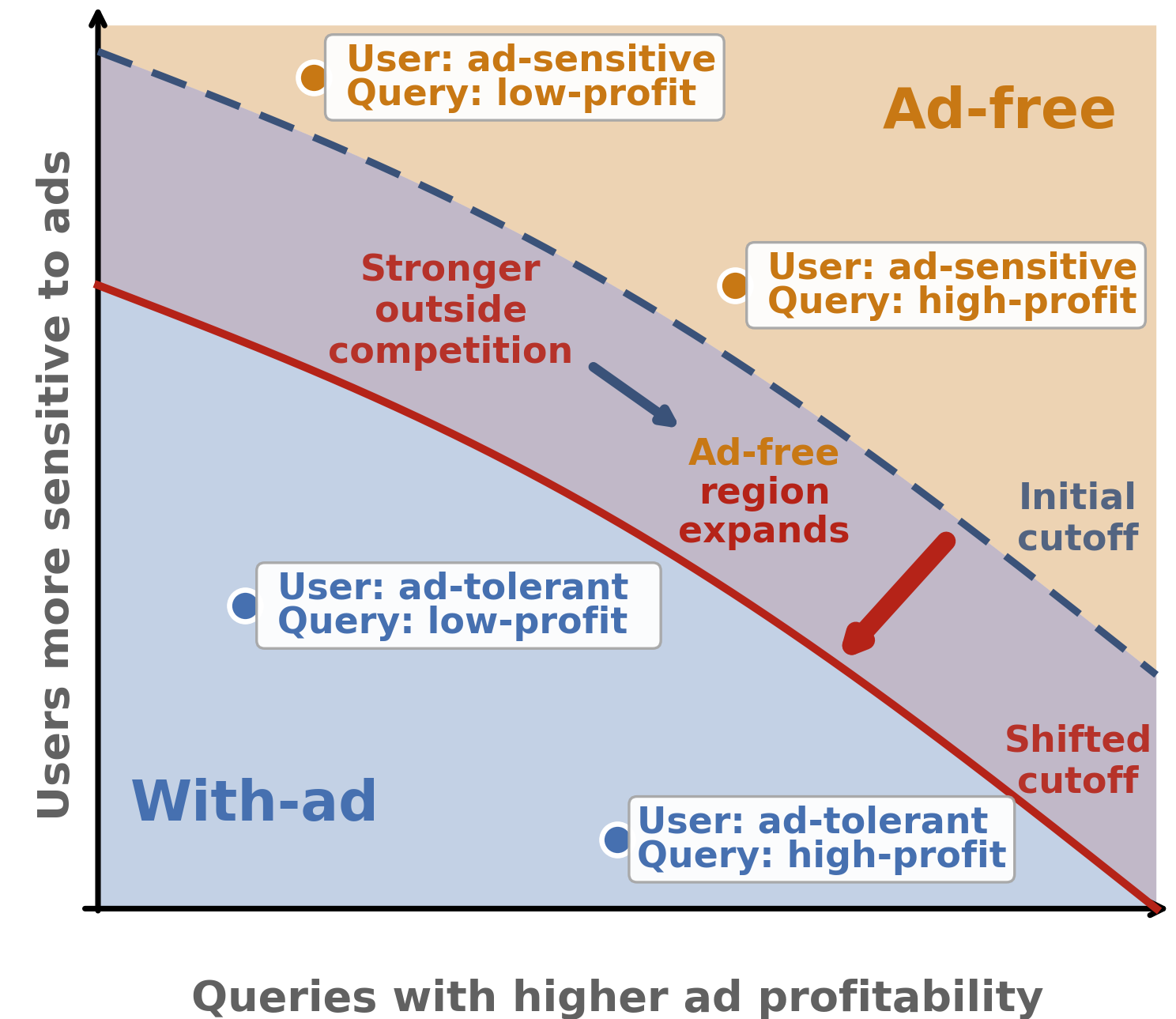}
  \caption{Optimal monetization policy thresholds for GEs. The dashed line shows the baseline cutoff between with-ad and ad-free responses; the solid line shows how outside competition shifts this cutoff, expanding the ad-free region.}
  \label{fig:intro_policy_overview}
  \vspace{-0.5cm}
\end{figure}

\section{Related Work}
\label{sec:related_work}

\textbf{Generative Engines and Monetization Design.}
Generative Engines synthesize cohesive responses from retrieved sources rather than ranking individual documents \citep{lewis2020retrieval,aggarwal2024geo}.
Unlike traditional web search---primarily monetized through ads and user clicks at comparatively low serving costs \citep{varian2007position,evans2008economics,athey2011position}---Generative Engines incur substantial inference costs \citep{kwon2023efficient} and are also associated with a sharp decline in downstream clicks \citep{arcintermedia2026case,ning2026agentic}.
This shift motivates LLM-native advertising embedded directly in generated responses.
Recent work proposes mechanism designs for ad auctions within retrieval-augmented generation streams \citep{hajiaghayi2024ad,dubey2024auctions,duetting2024mechanism,balseiro2025position}, with extensions to conversational and other generative auction formats \citep{mordo2024sponsored,bhawalkar2025sponsored,ma2025genauction,zhao2025llm}.


\textbf{User Behavior and Economic Model of Discrete Choice.}
Users' search engagement has traditionally been modeled through click and browsing models for ranked lists \citep{craswell2008experimental,dupret2008user,chapelle2009dynamic,joachims2017accurately}, with related work using signals such as dwell time and query reformulation to evaluate user satisfaction \citep{agichtein2006improving,li2009good}.
With AI-generated content, users increasingly obtain answers without clicking \citep{pew2025_ai_summary_clicks,kaiser2025new}. Engagement shifts away from link-based navigation \citep{Gleason2023SERPComponents,kirsten2025characterizing}, motivating a focus on substitution across response formats rather than click allocation alone.
Accordingly, economic models of random-utility discrete choice provides a natural foundation for modeling user choice among heterogeneous response formats \citep{mcfadden1972conditional,ben1985discrete}, consistent with its use in search-result substitution settings \citep{GhoseLi2011www, yao2011dynamic, GhoseLi2012RankingDesign, jerath2014consumer,jeziorski2015makes}.

\begin{figure*}[t]
\centering
\includegraphics[width=\linewidth]{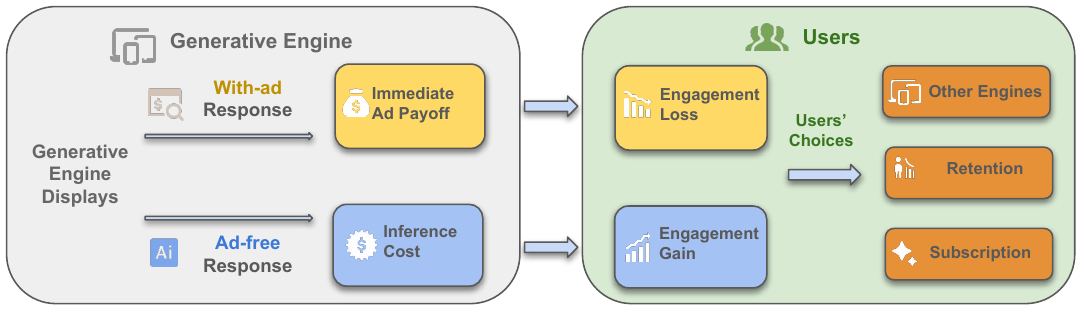}
\vspace{-15pt}
\caption{
Overview of the generative engine's monetization framework.
The generative engine chooses between a with-ad response (yielding immediate ad payoff) and an ad-free response (incurring inference cost).
Each response affects user engagement: with-ad responses reduces engagement, while ad-free responses builds it.
Users then choose among the generative engine, subscription, or outside alternatives, determining long-term retention and conversion.
}
\label{fig:1}
\vspace{-15pt}
\end{figure*}

\textbf{Dynamics, Retention, and Subscription Economics.}

User churn models connect short-term engagement to future retention and lifetime value \citep{gupta2006modeling,fader2009probability,fader2010customer,ascarza2013joint}.
Building on these dynamics, subscription economics emphasize renewal costs and experience accumulation \citep{dube2010state,della2006paying,farrell2007coordination,miller2023sophisticated}, including freemium and trial-based paths to paid upgrades \citep{einav2025selling,datta2015challenge,shi2019freemium}.
Improving the free-tier experience highlights a core tension: richer free offerings can erode short-term ad revenue \citep{taylor2013_cannibalization}.



\section{Economic Framework}
\label{sec:setup}

In this section, we introduce and formalize our dynamic generative engine monetization framework, which takes into account the ongoing interaction between a generative engine's design choices and its users' responses over time. 

We build the framework in four parts. \Cref{sec:timing} formally defines the general market setting with GEs and users. \Cref{sec:within_query} specifies how users respond to the generative engine's query-level design choice. \Cref{sec:dynamics} describes how users' per-query response accumulate their experiences over time, shaping their future engagement and subscription decisions. \Cref{sec:dp} formulates the generative engine's optimization problem over this repeated interaction.

\subsection{Framework Setup and Interaction.}
\label{sec:timing}

We begin by defining the basic elements of the framework, which includes generative engines, users, and their interaction sequence over time.

\textbf{Generative Engines.} Since current major generative engines offer differ tiers to its users, such as free, with-add, or paid subscriptions \citep{googleads2024ai-max-search-campaigns, openai_chatgpt_go_2026}, a central design decision in our framework is how the engine serves its non-subscribed users. Specifically, for each user query, the engine selects between two response formats: a with-ad response (\AIa) or an ad-free response (\AIf), both at no charge. Over time, a non-subscribed user may convert to a paid subscription at price $p$, after which they receive ad-free responses (\AIp) by default and the engine no longer faces a design choice for that user \citep{miller2023sophisticated,einav2025selling}.

\textbf{Users.} On the user side, each user has (i) a \textit{context}, which captures persistent user characteristics (e.g., demographics, usage patterns, and preferences), influencing how they engage with the generative engine and (ii)  \textit{state variables}, which captures their state at a given time. In particular, a user's context $x \in \mathcal{X}$, is drawn from a population distribution $\mu$ and remains stable over time, as is standard in user modeling \citep{mcfadden1972conditional, gupta2006modeling}. The user state variables, by contrast, evolve over time driven by their interactions with the generative engine. Details on user context and states are further discussed in \Cref{sec:dynamics}.

\textbf{Interaction sequence.} Next, given a generative engine and an user, we now describe their interaction. Fix a user context $x$. In each period $t=0,1,2,\ldots$:
\begin{enumerate}[noitemsep, topsep=0pt]    
    \item \textbf{Query and display.} The user generates and submits a query $q_t\in \mathcal Q$ (we assume queries are conditionally i.i.d. given user context, i.e., $q_t \sim \mathcal{P}(\cdot \mid x)$ \footnote{The i.i.d assumption can be relaxed without changing the framework's structure. More general stationary query dynamics can be handled by augmenting the state (e.g., including $q_t$). We assume i.i.d.\ queries to keep the state minimal; modeling query transitions is orthogonal to our main results. See \Cref{app:framework_notes} for details.}) to the generative engine. 
    
    \item \textbf{Generative Engine Response.} If the user is active and non-subscribed, the generative engine chooses a response format $A_t \in \{\AIa, \AIf\}$.

    \item \textbf{User engagement.} The user either engages with the 
    displayed response or takes an outside alternative such as other 
    competing GEs. This choice determines the generative engine's 
    immediate payoff (see \Cref{sec:within_query}).

    \item \textbf{Cross-period dynamics: retention and subscription.} 
    The engagement updates the user’s state variables.
    Depending on the updated state, the user may convert to a paid subscription (\AIp) according to the conversion rule in \Cref{sec:dynamics}. 
\end{enumerate}

Finally, the interaction repeats: the updated state variables and subscription decision jointly determine the user's continued engagement with the generative engine in the next period (see \Cref{sec:dynamics}).

\subsection{User Per-query Choice and generative engine's Payoffs}\label{sec:within_query}

Building on the interaction model above, we now characterize the user engagement towards the generative engine, and the resulting payoff to the engine.

\textbf{User per-query choice.}
Given a user context $x$, a query $q$, and a response $A_t=a$ displayed by the generative engine, the user decides whether to engage with the response $a$ or not. We denote the decision variable as $Y_t \in {1,0}$, where $Y_t = 1$ denotes reading and using the generated content, and $Y_t = 0$ denotes taking an outside alternative, such as switching to a competing generative engine.

The user's engagement is driven by two utilities: that of the generative engine's displayed response, $v_a(x,q)$, and that of the outside alternative, $v_0(x,q)$.
The engagement probability follows a binary logit, widely adopted in discrete-choice modeling \citep{luce1959individual,mcfadden1972conditional,ben1985discrete}:
\begin{align}
& \Pr(Y_t=1\mid x,q,a)
=
\frac{\exp\!\big(v_a(x,q)\big)}{\exp\!\big(v_a(x,q)\big)+\exp\!\big(v_0(x,q)\big)},
\nonumber
\\
& \Pr(Y_t=0\mid x,q,a)=1-\Pr(Y_t=1\mid x,q,a).
\label{eq:within_query_logit}
\end{align}

\textbf{Generative engine payoffs.} Given a generative engine's response $a$ to user $x$ on query $q$, its expected per-period payoff is  denoted by $r_a(x,q)$, which is an expectation taken over the engagement decision $Y_t$. That is,
\begin{equation}
r_a(x,q)
=
\begin{cases}
\begin{aligned}
-\kappa +& R(x,q)\,\\
&\Pr(Y_t=1 \mid x,q,a)
\end{aligned}
& a=\AIa,\\
-\kappa
& a=\AIf,
\end{cases}
\label{eq:expected_payoff}
\end{equation}
where $R(x,q)$ is the expected ad revenue \emph{conditional on} engagement, and the choice probability is given by \eqref{eq:within_query_logit}. In addition, both response formats incur an inference cost $\kappa$ \citep{luccioni2023estimating} \footnote{ In practice, the inference cost $\kappa$ can also depend on user context and query. Our framework and theoretical results extend to the cost $\kappa(x,q)$ (see \Cref{app:kappa_xq} for details).}. Only a response with ads, \AIa, generates immediate ad revenue upon engagement. In contrast, displaying \AIf\ generates no immediate revenue; its value comes from retaining users and building toward subscription, as formalized in \Cref{sec:dynamics}.

\subsection{User Cross-Query Dynamics and Subscription}
\label{sec:dynamics}

Next, we define the state variables that track a user's past interactions with the generative engine. These variables shape the user's future retention and subscription decisions.

\textbf{State definition and update process.}
Each user's interaction history is summarized by three state variables: (1) an AI-experience state $S_t$, which captures cumulative experience from interactions with \textit{ad-free} responses; (2) an ad-exposure state $C_t$, tracking cumulative exposure to \textit{with-ad} responses; and (3) a subscription indicator $Z_t \in {0,1}$ ($Z_t=1$ means subscribed to \AIp). Formally, the user enters period $t$ with state:
$$\quad (S_t, C_t, Z_t) \in [0,\infty) \times [0,\infty) \times {0,1}.$$

When $Z_t=0$ (i.e., the user is unsubscribed), the generative engine chooses $A_t\in\{\AIa,\AIf\}$ for query $q_t$, and the two user experience states $S_t$ and $C_t$ update only upon user engagement (i.e., $Y_t=1$). Engaging with an ad-free response \AIf\ increments $S_t$, while engaging with a with-ad response \AIa\ increments $C_t$. 
Formally,
\begin{align}
S_{t+1}
&=
S_t + \mathbf{1}\{A_t=\AIf, Y_t = 1\}\,\delta_S(x,q_t,S_t),
\nonumber
\\
C_{t+1}
&=
C_t + \mathbf{1}\{A_t=\AIa, Y_t = 1\}\,\delta_C(x,q_t, C_t),
\nonumber
\end{align}
where $\delta_S,\delta_C:\mathcal X\times\mathcal Q\times\mathbb R_+\to[0,\infty)$.
If the user chooses the outside option ($Y_t=0$), neither state changes.

\textbf{Conversion to the subscription.}
Over time, a user may convert to the paid subscription \AIp.
Conversion occurs once the user's accumulated AI experience $S_t$ exceeds a threshold that depends on user context $x$, subscription price $p$, and the ad exposure $C_{t}$, consistent with experience-driven adoption in freemium markets \citep{einav2025selling, miller2023sophisticated}.
Formally,
\begin{equation}
Z_{t} 
= 
\mathbf{1}\{Z_{t-1} = 1 \text{ or } S_{t} \geq \tau(x, p, C_{t})\},
\label{eq:conversion_rule}
\end{equation}
where $\tau(x,p,c)$ is the subscription threshold. 
The threshold is decreasing in $c$, reflecting that cumulative ad exposure can accelerate subscription demand \citep{cho2004people, anderson2005market}. 
\Cref{app:prob_adoption} replaces this threshold format with a smooth adoption hazard without changing the analysis.

We treat subscription as permanent: if $Z_t = 1$, then $Z_{t'} = 1$ for all $t' \geq t$. After subscription, the generative engine earns a per-period net payoff 
$r_{\mathrm{sub}}(p) := p - \kappa_{\mathrm{sub}}$, 
where $\kappa_{\mathrm{sub}} \geq \kappa$ is the inference cost for subscription-tier. 
This setup is standard in subscription markets and keeps the dynamic program tractable \citep{einav2025selling,della2006paying,farrell2007coordination}. 
\Cref{app:churn} relaxes this by allowing post-subscription churn without changing the main theoretical arguments.

\textbf{Cross-period retention and termination.}
Whether a user remains active depends on their accumulated interactions with the generative engine. 
After updating $(S_{t+1}, C_{t+1}, Z_{t+1})$, the user remains active in period $t+1$ with probability
$\rho_x(S_{t+1}, C_{t+1}, Z_{t+1}) \in [0,1]$,
where $\rho_x(s,c,1)=1$ for subscribed users. Thus, subscribed users remain active.
Retention increases with AI experience $S_{t+1}$ and decreases with ad exposure $C_{t+1}$, reflecting that positive interactions encourage continued use while excessive ad exposure drives users away \citep{gupta2006modeling, ascarza2013joint}. We model churn as permanent, which is standard in customer-base analysis \citep{fader2009probability, fader2010customer}.

\subsection{Generative Engine's Objective and Design Policy}
\label{sec:dp}

Given the user’s per-query response and the cross-period dynamics defined above, the generative engine now faces a dynamic decision problem: how to choose response formats over time to maximize long-run value.

\textbf{Policy class and value function.} For pre-subscription users, the generative engine's policy maps the user's current state and query to a response format. 
Formally, a stationary Markov policy is a measurable map
\[
g_x:\ [0,\infty)\times[0,\infty)\times\mathcal Q \to \{\AIa,\AIf\}, g_{x} \in \mathcal G_x,
\]

The value function $V_x(s,c,z)$ is the generative engine's maximum expected discounted payoff starting from state $(S_0, C_0, Z_0) = (s,c,z)$:
\begin{align}
V_x(s,c,z)
&:= 
\sup_{g_x\in\mathcal G_x}\; 
\mathbb E^{g_{x}}
\!\left[\sum_{t\ge 0}\beta^t\Big(\mathbf 1\{Z_t=0\}\, r_{A_t}(x,q_t)\right.
\nonumber
\\
&\left.\quad \quad \quad \quad +\mathbf 1\{Z_t=1\}\, r_{\sub}(p)\Big)\right],
\label{eq:value_def}
\end{align}
where future payoffs are discounted by $\beta \in (0,1)$, $A_t \in {\AIa, \AIf}$ is the response format chosen by $g_x$ when $Z_t = 0$, and $r_{A_t}(x,q_t)$ is the expected per-period payoff from \eqref{eq:expected_payoff}. Since per-period payoffs are bounded and $\beta < 1$, the value function $V_x(s,c,z)$ is well-defined under standard contraction arguments, ensuring the Bellman equations below are well-posed.

\textbf{Bellman equations.}
For a subscribed user ($z=1$), the generative engine makes no display decision and receives $r_{\sub}(p)$ each period, implying
\begin{equation}
V_x(s,c,1)=r_{\sub}(p)+\beta V_x(s,c,1)=\frac{r_{\sub}(p)}{1-\beta}.
\label{eq:bellman_subscribed_compact}
\end{equation}

For a non-subscribed user ($z=0$), conditional on $(x,s,c)$ and a query $q$, the generative engine chooses $a \in \{\AIa, \AIf\}$. 
The user's within-query choice $Y \in \{1,0\}$ follows \eqref{eq:within_query_logit}, and the one-period payoff is $r_a$ as defined in \eqref{eq:expected_payoff}. 
For each query, the generative engine compares the immediate ad revenue from \AIa against the discounted continuation value from \AIf, which generates no immediate revenue but boosts retention and subscription conversion. 
Using the state update rules from \Cref{sec:dynamics}, we define the post-update states:
\begin{align}
S_{a,y}^+(x,q,s)
&:= s + \mathbf 1\{a=\AIf, y = 1\}\delta_S(x,q,s),
\nonumber
\\
C_{a,y}^+(x,q,c)
&:= c + \mathbf 1\{a=\AIa, y = 1\}\delta_C(x,q,c),
\nonumber
\\
Z_{a,y}^+(x,q,s,c)
&:= \mathbf 1\{\,S_{a,y}^+(x,q,s)\ge \tau(x,p,C_{a,y}^+(x,q,c))\,\}.
\label{eq:post_update_mappings_dp}
\end{align}
Let $Y\sim \Pr(\cdot\mid x,q,a)$ denote the within-query choice, and write
$(S^+,C^+,Z^+):=(S_{a,Y}^+(x,q,s),\,C_{a,Y}^+(x,q,c),\,Z_{a,Y}^+(x,q,s,c))$.
The pre-subscription Bellman equation is
\begin{align}
V_x(s,c,0)
&=
\mathbb E_{q\sim \mathcal P(\cdot\mid x)}
\Bigg[
\max_{a\in\{\AIa,\AIf\}}
\Big\{
r_a(x,q)
\nonumber
\\
& \quad \quad +
\beta\,
\mathbb E_{Y\sim \Pr(\cdot\mid x,q,a)}
\Big[
\rho_{x} \big(S^+,C^+,Z^+\big)
\nonumber
\\
& \quad \quad \quad \cdot\, V_x \big(S^+,C^+,Z^+\big)
\Big]
\Big\}
\Bigg],
\label{eq:bellman_presub_compact}
\end{align}
\noindent

\textbf{Connection to a static benchmark.}
When cross-period interactions are absent, the design problem simplifies: each query is an independent decision with no future consequences. Specifically, the dynamic program reduces to a per-query static comparison between \AIa\ and \AIf\ when (i) retention is exogenous (independent of accumulated experience) and (ii) subscription has no option value (e.g., $\tau = \infty$ or $r_{\sub} = 0$). In this case, the generative engine simply chooses whichever response format yields higher immediate expected payoff on each query.

\section{Theoretical Results}
\label{sec:dp-results}
This section characterizes which users and queries should receive with-ad versus ad-free responses, and how this allocation shifts across market conditions. 
Understanding this variation helps identify the economic conditions under which different response formats are optimal.
We first derive the generative engine's optimal display policy, then study how it varies with user and query types (\Cref{sec:dp-cutoff}) and market conditions, including AI economics for inference cost and subscription price (\Cref{sec:dp-metrics}) and outside competition (\Cref{sec:dp-outside}).

\subsection{Optimal Policy Characterization}

We show that the optimal display policy takes a comparison rule: for a given user and query, the generative engine displays \AIa\ if and only if its total value — immediate ad revenue plus long-run value from retention and subscription — exceeds that of \AIf.

A policy is optimal if it maximizes the value function $V_x(s,c,0)$ in \eqref{eq:value_def}. 
For analysis, we impose standard regularity conditions (e.g., boundedness, measurability) throughout, which are common in dynamic programming and discrete-choice models \citep{puterman2014markov, rust1987optimal}; see \Cref{app:regularity} for details.
The following proposition characterizes this optimal policy.

\begin{proposition}[Optimal policy as a comparison rule]\label{prop:opt_gate}

Under \Cref{assm:regularity}, for any query $q$ and pre-subscription state $(s,c,0)$, define $Q_x^a(s,c,q)$ as the generative engine's expected payoff from displaying response $a \in {\AIa, \AIf}$, including the discounted continuation value:
\begin{align}
& Q_x^{a}(s,c,q)
\nonumber
\\
& \quad :=
r_{a}(x,q)
+\beta \mathbb E_{Y\sim \Pr(\cdot\mid x,q,a)}\!\big[ 
\nonumber
\\
& \quad \quad \; \;  \rho_{x}\big(S^+_{a,Y},C^+_{a,Y},Z^+_{a,Y}\big)\,
V_x \big(S^+_{a,Y},C^+_{a,Y},Z^+_{a,Y}\big) \big],
\end{align}
which is the expected payoff from choosing action $a$ on query $q$, including discounted continuation value.

Define \AIa's value edge over \AIf as $\Delta_x(s,c,q):=Q_x^{\AIa}(s,c,q)-Q_x^{\AIf}(s,c,q)$. Then the \textit{optimal policy} $g_{x}$ chooses \AIa (i.e., $g_{x}^{*}(s,c,q)=\AIa$) if and only if $\Delta_x(s,c,q)\ge 0$.
See \Cref{app:proof}. 
\qed
\end{proposition}

Next, we decompose \AIa's value edge over \AIf $\Delta_x(s,c,q)$ into short-term monetization and long-term value from retention and conversion.

\begin{remark}[Short-term vs.\ long-term trade-off]\label{rem:tradeoff}
For any pre-subscription state and query $(x,s,c,q)$, we can identify the generative engine's short-term and long-term value trade-off by expanding the value edge $\Delta_x(s,c,q)$. 
We denote the long-term value function as

The value edge $\Delta_x(s,c,q)$ decomposes into two components: a short-term ad revenue edge and a long-term value edge.
Define the expected future value from retaining the user as
\begin{align}
&W_x(S^+_{a,Y},C^+_{a,Y},Z^+_{a,Y})
\\
\nonumber
&\quad :=
\rho_x(S^+_{a,Y},C^+_{a,Y},Z^+_{a,Y})\,V_x(S^+_{a,Y},C^+_{a,Y},Z^+_{a,Y}).
\end{align}
Then:
\begin{align}
&\Delta_x(s,c,q)
=
\nonumber
\\
& \underbrace{R(x,q)\,\Pr(Y=1\mid x,q,\AIa)}_{\text{short-term monetization edge of \AIa over \AIf}}
\nonumber\\
& +
\beta\Big(
\mathbb E_{Y\sim \Pr(\cdot\mid x,q,\AIa)}\big[ W_x(S^+_{\AIa,Y},C^+_{\AIa,Y},Z^+_{\AIa,Y}) \big]
\nonumber
\\
& \underbrace{
-
\mathbb E_{Y\sim \Pr(\cdot\mid x,q,\AIf)}\big[ W_x(S^+_{\AIf,Y},C^+_{\AIf,Y},Z^+_{\AIf,Y}) \big]
\Big)}_{\text{long-term value edge (retention and subscription)}}.
\label{eq:delta_tradeoff}
\end{align}
The first component is \AIa's expected ad revenue. 
The second component captures the long-term trade-off: \AIf\ builds AI experience and drives subscription conversion, whereas \AIa\ raises ad exposure that may reduce future engagement. \qed
\end{remark}

\Cref{fig:sim_threshold} illustrates this decomposition.
As AI experience $s$ increases, the long-term component increasingly favors \AIf\ through higher retention and earlier subscription conversion (\Cref{fig:sim_threshold}, left).
Once $s$ exceeds the subscription threshold $\tau$, the user subscribes and the value reaches the subscribed level (\Cref{fig:sim_threshold}, right).

\subsection{Implications of User and Query Types}
\label{sec:dp-cutoff}

We now characterize how the optimal design policy varies across user and query types, identifying when the policy switches between \AIa and \AIf.

\textbf{User and query types.}
We capture variation in the generative engine’s environment using user and query types.
Each user has type $\gamma(x)$ capturing \emph{ad sensitivity}, i.e., the user’s responsiveness to ads---how strongly additional \textit{ad exposure} lowers future engagement \citep{anderson2005market,cho2004people,edwards2002forced}.

On the query side, $r_{\AIa}(x,q)$ measures \emph{ad profitability} (the expected short-term revenue from showing ads on query $q$) \citep{ghose2009empirical}.
Further, $\psi(x,q)$ measures the \emph{AI experience gain} from showing an ad-free generative response, which accelerates user's experience accumulation and subscription conversion.
Formally, these types enter the model through the retention and conversion functions $(\rho,\tau)$ and the experience updates $\delta_S, \delta_C$; see \Cref{app:regularity} for detailed assumptions.

\begin{figure}[t]
  \centering
  \includegraphics[width=\linewidth]{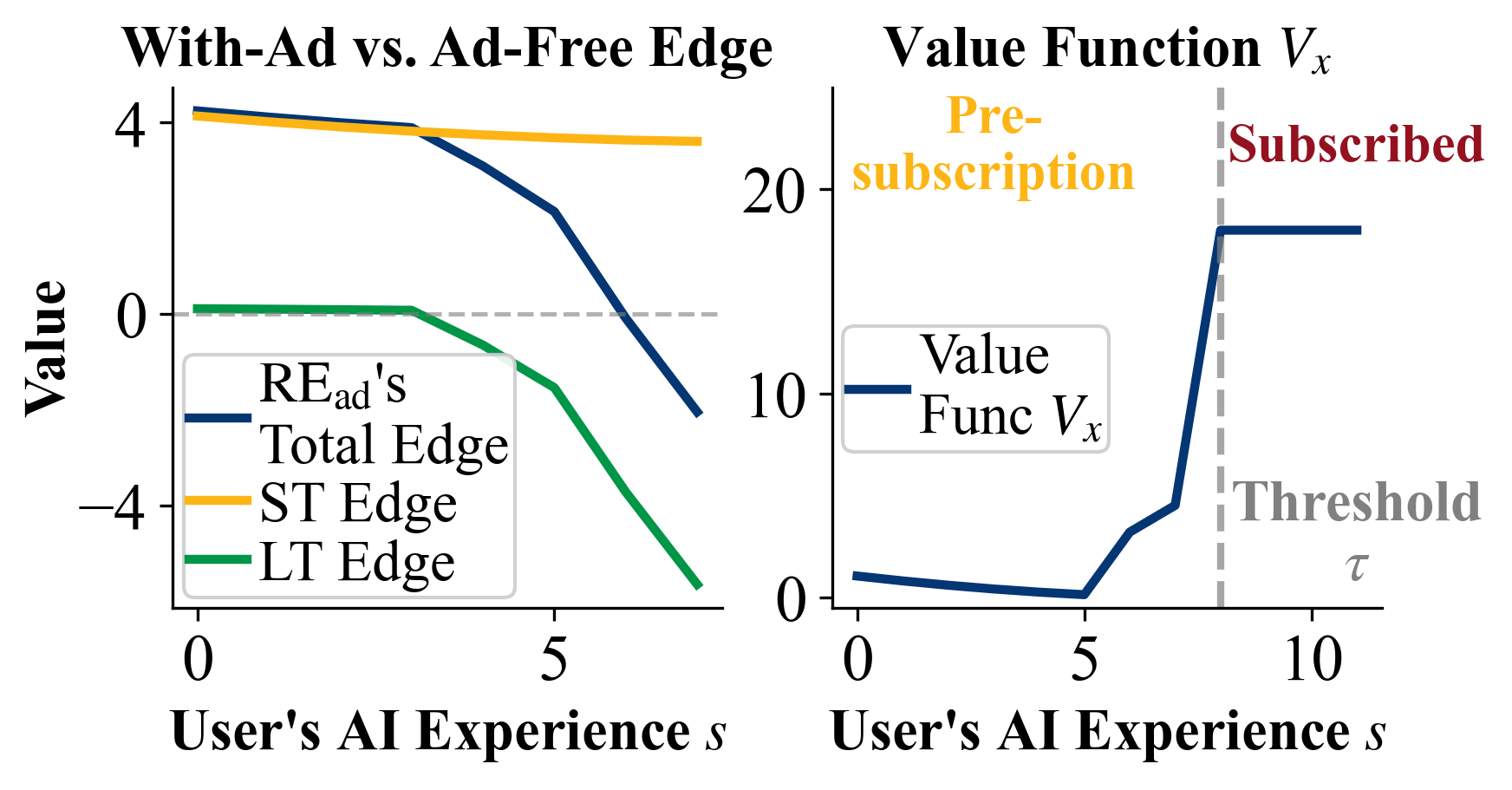}
  \vspace{-15pt}
  \caption{
  Simulation of the optimal policy in \Cref{prop:opt_gate} for a fixed user and ad exposure state, as AI experience increases.
  \emph{Left:} \AIa's total edge over \AIf\ and its short-term (ST) and long-term (LT) components from \Cref{rem:tradeoff}.
  \emph{Right:} Pre-subscription value function $V_{x}$; the dashed line marks the subscription threshold, beyond which the value equals the subscribed level.
  }
  \label{fig:sim_threshold}
  \vspace{-15pt}
\end{figure}

\begin{proposition}[Threshold structure over user and query features]\label{prop:type_order}
Under \Cref{assm:regularity,assm:type_mono}, fix a pre-subscription state $(s,c,0)$.
Define the user--query type tuple as
\[
t\equiv(\gamma,r,\psi):=\big(\gamma(x),\,r_{\AIa}(x,q),\,\psi(x,q)\big).
\]
Define a partial order on types by
\[
(\gamma,r,\psi)\preceq(\gamma',r',\psi')
\Longleftrightarrow
\gamma\le\gamma',\ \psi\le\psi',\ \text{and } r\ge r'.
\]
Then the following hold:
\begin{enumerate}[noitemsep, topsep=0pt]
\item \textbf{(Monotone ordering across user--query types).}
If \AIf is optimal at type $t$, then it is also optimal at any higher type $t'\succeq t$.

\item \textbf{(Type thresholds along a single index).}
Holding $(\gamma,r)$ fixed, there exists a threshold
$\psi^*(s,c;\gamma,r)\in[-\infty,\infty]$ such that
\[
\AIa \text{ is optimal } \Longleftrightarrow \psi \le \psi^*(s,c;\gamma,r).
\]
Holding $(\gamma,\psi)$ fixed, there exists a threshold
$r^*(s,c;\gamma,\psi)\in[-\infty,\infty]$ such that
\[
\AIa \text{ is optimal } \Longleftrightarrow r \ge r^*(s,c;\gamma,\psi).
\]
\end{enumerate}
\end{proposition}


\begin{figure}[t]
  \centering
  \includegraphics[width=\linewidth]{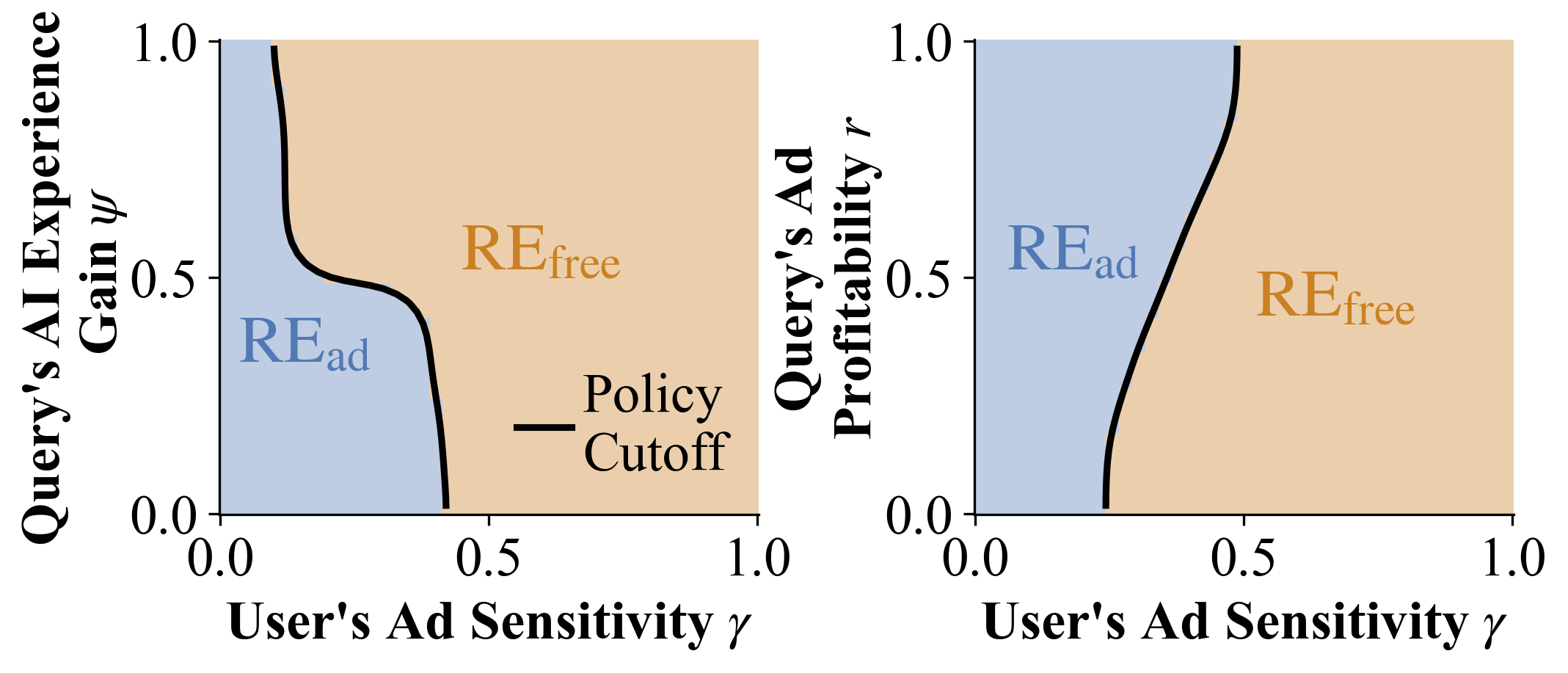}
  \caption{
  Illustration of the optimal policy cutoff across user and query types from \Cref{prop:type_order}.
  \emph{Left:} User \textit{ad sensitivity} $\gamma$ vs.\ query \textit{AI experience gain} $\psi$.
  \emph{Right:} User \textit{ad sensitivity} $\gamma$ vs.\ query \textit{ad profitability} $r$.
  The solid curve separates \AIa\ (blue) and \AIf\ (orange) regions.
  }
  \label{fig:type_pdf_user}
\end{figure}

\begin{remark}[Interpretation: when to show \AIf and \AIa]\label{rem:type_regimes}
\Cref{prop:type_order} delivers a key ordering result: the generative engine decides toward \AIf for users who are more \textit{ad-sensitive} and for queries where an ad-free response generates larger user experience gains, and shifts toward \AIa when the query’s immediate \textit{ad revenue} is higher.
Moreover, this ordering implies a clear cutoff: as a user or query becomes more favorable to \AIf\ along these types, there is a point at which the generative engine switches from showing \AIa to \AIf.
\end{remark}

\Cref{fig:type_pdf_user} presents the type-ordered structure in \Cref{prop:type_order} by overlaying the optimal display decision on the population density of user-type indices.
Holding the other parameters fixed, the \AIf region expands monotonically with ad sensitivity $\gamma$, yielding a clean cutoff that splits the population into \AIa-preferred and \AIf-preferred regions.
See the illustration over query types $(r_{\AIa},\psi)$ in \Cref{app:extra_simu_user_query}. We also show the illustration of \Cref{fig:sim_threshold} under user heterogeneity at \Cref{app:extra_simu_policy}.

\textbf{Generative engine’s weight on future outcomes.}
We also show that a more forward-looking generative engine (higher discount factor $\beta$) displays \AIf more often, as \AIf improves long-term value through retention and experience accumulation (\Cref{app:dp-horizon-beta}).

\subsection{Implications of AI Economics: Inference Cost and Subscription Price}
\label{sec:dp-metrics}

Inference costs and subscription pricing jointly influence the allocation between \AIa\ and \AIf. Inference costs affect the relative value of keeping a user in the free tier versus moving that user into the paid tier, while subscription pricing changes the revenue from conversion.

\textbf{Inference cost.} 
As defined in \Cref{sec:within_query}, both \AIa\ and \AIf\ incur the same free-tier inference cost $\kappa$ per query, while the paid subscription tier carries cost $\kappa_{\sub} \geq \kappa$.
To analyze how these costs affect the design policy, consider a user close to subscribing, i.e., whose AI experience $s$ is near the subscription threshold $\tau(x,p,c)$, where inference costs have the most direct effect on the design choice (see \Cref{assm:one_step} for the formal setup).

\begin{proposition}[Free-tier and paid-tier inference costs have opposite effects near subscription]
\label{prop:cs-cost-local}
Under \Cref{assm:regularity,assm:cs-cost,assm:one_step}, recall the \AIa's value edge $\Delta_x(s,c,q)$ from \Cref{prop:opt_gate}; here we index it by inference costs and write $\Delta_x(s,c,q;\kappa,\kappa_{\sub})$.

\begin{enumerate}[noitemsep, topsep=0pt]
\item \textbf{Higher free-tier cost $\kappa$ raises incentives to display \AIf in the near-threshold region.}
$\Delta_x(s,c,q;\kappa,\kappa_{\sub})$ is weakly decreasing in $\kappa$. That is, at near-threshold states where displaying \AIf\ triggers subscription while displaying \AIa\ leaves the user pre-subscription, increasing $\kappa$ pushes the decision toward \AIf.

\item \textbf{Higher paid-tier cost $\kappa_{\sub}$ reduces incentives to display \AIf in the near-threshold region.}
Under \Cref{assm:one_step}, $\Delta_x(s,c,q;\kappa,\kappa_{\sub})$ is weakly increasing in $\kappa_{\sub}$. That is, at near-threshold states where displaying \AIf would trigger subscription, increasing $\kappa_{\sub}$ pushes the decision further toward \AIa. 

\end{enumerate}
\end{proposition}

Near the subscription threshold, the two costs work in opposite directions. A higher free-tier cost makes it more attractive to move the user out of the free tier, while a higher paid-tier cost makes immediate conversion less attractive.

\begin{figure}[t]
  \centering
  \includegraphics[width=\linewidth]{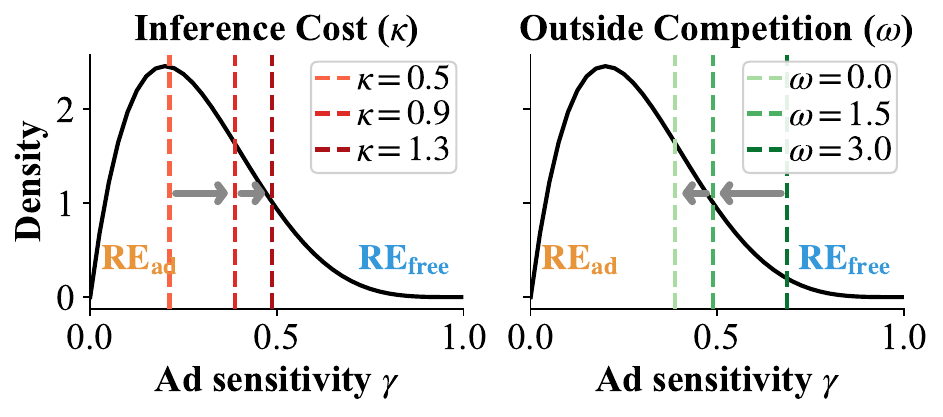}
  \caption{
  Policy cutoff shift illustration under inference cost and outside competition, at a fixed pre-subscription state $(s,c)$.
  Ad sensitivity follows $\gamma\sim\mathrm{Beta}$.
  The solid curve shows the user-type density; users to the left of the cutoff are served \AIa, while those to the right are served \AIf.
  Dashed lines indicate the cutoff $\gamma^*$ at each parameter level, and arrows indicate the direction of the cutoff shift.
  \emph{(a)}~Varying inference cost $\kappa$; \emph{(b)}~varying outside competition $\omega$.
  }
  \vspace{-15pt}
  \label{fig:competition_cutoff}
\end{figure}

\textbf{Subscription price.}
Another key economic factor is the subscription price $p$. 
As defined in \Cref{sec:dynamics}, $p$ enters through both the subscription conversion threshold $\tau(x,p)$ and the post-conversion revenue $r_{\sub}(p)$, so it affects the generative engine’s incentive to display \AIf. 
We study the price effect in the same one-step conversion region (as in \Cref{prop:cs-cost-local}), where displaying \AIf\ triggers immediate subscription.

\begin{proposition}[Higher subscription price shifts the design toward \AIf near the threshold.]\label{prop:cs-p-local}
Under \Cref{assm:regularity,assm:cs-cost,assm:cs-p-nuanced}, fix a pre-subscription decision point $(x,s,c,q)$ and consider any price interval
$I$ such that for all $p\in I$, the one-step conversion conditions in \Cref{assm:one_step} hold at $(x,p,s,c,q)$.
Then the \AIa's value edge $\Delta_x(s,c,q;\,p)$ is weakly decreasing in $p$ on $I$.
Equivalently, in this near-threshold region, a higher subscription price weakly shifts the optimal design toward \AIf.
\end{proposition}

These propositions share a common insight: AI economics has the sharpest implications for users who are close to subscribing. 
For these near-threshold users, showing an ad-free response can easily push them over the subscription threshold. 
In this region, a higher free-tier cost $\kappa$, a lower paid-tier cost $\kappa_{\sub}$, and a higher subscription price $p$ all make it more attractive to move the user into the paid tier by displaying \AIf. 
Away from the near-threshold region, price can also shift the conversion threshold $\tau(x,p)$ but the overall effect of $p$ can be more nuanced.

\subsection{Impact of Outside Competition}
\label{sec:dp-outside}

When outside alternatives strengthen, users engage less with displayed responses, reducing the generative engine's short-term ad revenue and shifting the long-term comparison toward \AIf.

\textbf{Competitive-pressure channel.}
In our framework (\Cref{sec:setup}), outside competition enters through the outside-option utility $v_0(x,q)$ in the per-query choice model, capturing factors such as advances in rival LLMs, ad-free policies by competitors \citep{ChmielewskiSeetharamanCherney2026AnthropicSuperBowlAds}, or the growth of alternative GEs \citep{Singh2025Perplexity780MillionQueries}.

To formalize how outcomes change as the outside option improves, we index a uniform increase in outside-option attractiveness by $\omega\ge 0$, shifting $v_0(x,q)\to v_0(x,q)+\omega$ while holding all other parameters fixed. The uniform shift is a tractable baseline; heterogeneous shifts across user--query pairs are a natural extension.
As outside competition intensifies, users disengage more from the generative engine, directly reducing \AIa's short-term ad revenue and shifting the long-term value advantage toward \AIf.
To capture this full effect, we allow the continuation values to depend on $\omega$ and impose a sufficient condition on the long-term value advantage of \AIf over \AIa. Under this condition, \AIa's total value edge shrinks as outside competition grows, shifting the optimal policy toward \AIf (see \Cref{assm:outside_comp} in \Cref{app:regularity} for details).

\begin{figure*}[t]
  \centering
  \includegraphics[width=\linewidth]{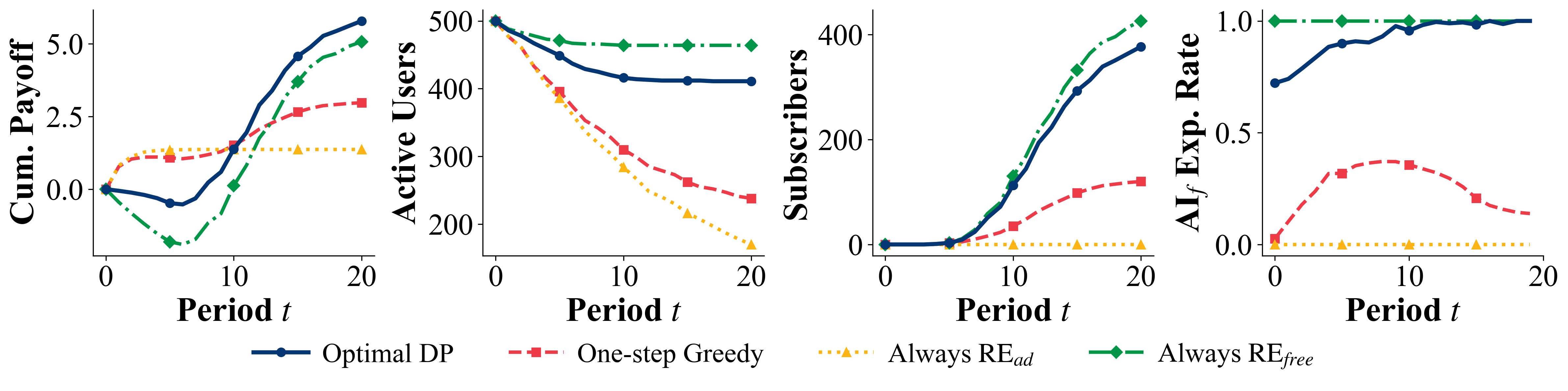}
  \vspace{-15pt}
  \caption{
Comparison of four policies: \textit{Optimal DP}, \textit{One-step greedy}, \textit{Always \AIa}, and \textit{Always \AIf}; $N=500$, $T=20$, $\beta=0.95$.
\emph{(1)} Discounted cumulative generative engine payoff up to $t$.
\emph{(2)} Total active users (non-subscribers plus subscribers).
\emph{(3)} Cumulative subscribers by $t$.
\emph{(4)} AI exposure rate: the share of \emph{active non-subscribers} shown \AIf\ at time $t$.
}
  \label{fig:experiment_time}
  \vspace{-15pt}
\end{figure*}

\begin{proposition}[Intensified competition shifts the policy toward \AIf]\label{prop:outside_competition}
Under \Cref{assm:regularity,assm:outside_comp}, for any fixed pre-subscription decision point $(s,c,q)$,
the \AIa's edge $\Delta_x(s,c,q;\omega)$ is weakly decreasing in outside-option intensity $\omega$.
Equivalently, stronger outside competition makes the optimal policy weakly more likely to select \AIf.
\end{proposition}

To understand the mechanism, recall the short-term vs.\ long-term decomposition of $\Delta_x$ from \Cref{rem:tradeoff}. Indexing both components by $\omega$, competition weakens \AIa through two channels: it directly reduces the share of users who engage with responses that include ads (shrinking the short-term ad-revenue component), and it widens \AIf's long-term value advantage (the long-term component), because ad-free experience accumulation becomes relatively more valuable when users are harder to retain. Both effects are formalized in \Cref{assm:outside_comp} and proved in \Cref{app:proof}.

\Cref{prop:outside_competition} tells us that competition shifts the \emph{overall} policy toward \AIf, but a natural follow-up question is: \emph{which} users and queries are affected first?
Combining \Cref{prop:outside_competition} with the type-threshold structure from \Cref{prop:type_order} --- under the additional requirement that \Cref{assm:type_mono} holds at each competition level $\omega$ (see \Cref{app:regularity}) --- we can trace how the cutoffs from \Cref{sec:dp-cutoff} move as competition intensifies:

\begin{proposition}[Competition shifts the type cutoffs toward \AIf]\label{prop:competition_type_shift}
Under \Cref{assm:regularity,assm:type_mono,assm:outside_comp}, with \Cref{assm:type_mono} holding for each fixed $\omega$, fix a pre-subscription state $(s,c,0)$ and index the type cutoffs from \Cref{prop:type_order} by $\omega$:
\begin{enumerate}[noitemsep, topsep=0pt]
\item \textbf{Experience-gain cutoff decreases.}
$\psi^*(s,c;\gamma,r,\omega)$ is weakly decreasing in $\omega$: under stronger competition, \AIf becomes optimal even for queries whose ad-free responses generate smaller experience gains.
\item \textbf{Ad-revenue cutoff increases.}
$r^*(s,c;\gamma,\psi,\omega)$ is weakly increasing in $\omega$: under stronger competition, showing ads is justified only when the query's immediate ad revenue is higher.
\item \textbf{Distributional consequence.}
For any fixed distribution $\mu$ over user--query types $(\gamma,r,\psi)$, the share of user--query pairs for which \AIa is optimal is weakly decreasing in $\omega$.
\end{enumerate}
\end{proposition}

In words, as competitors improve, the set of users and queries where showing ads is worthwhile shrinks: the generative engine can justify ads only for queries with high immediate revenue and low experience-gain potential, and the overall share of interactions served with ads decreases.

\begin{remark}[Subscription insulation]\label{rem:subscription_insulation}
Under the permanent-subscription benchmark, the subscribed-state value $V_x(s,c,1)=r_{\sub}(p)/(1-\beta)$ is independent of $\omega$, since subscribers consume \AIp directly without facing the outside option.
By contrast, the pre-subscription phase is fully exposed to competition through the per-query engagement channel.
This asymmetry suggests that competition selectively weakens free-tier value, helping explain why the generative engine's incentive to use \AIf as a bridge toward subscription grows stronger---particularly for users near the conversion threshold $\tau(x,p)$.
\end{remark}

\Cref{fig:competition_cutoff} illustrates this cutoff shift numerically: as outside competition strengthens from $\omega=0$ to $\omega=1$, the policy threshold $\gamma^*$ moves steadily to the right, shrinking the share of users served ads.

\textbf{Implication for competitive markets.}
These results indicate that ad-heavy strategies are most viable when outside competition is weak but become less sustainable as rival GEs improve.
The generative engine is pushed toward ad-free display to retain users who would otherwise switch, and toward subscription as a revenue source insulated from competition---consistent with the observed trend of generative engines investing in premium tiers as the market grows more competitive \citep{einav2025selling}.
This competitive-pressure effect complements the discount-factor result in \Cref{app:dp-horizon-beta}: both outside competition and forward-looking orientation independently push the policy toward \AIf, but through distinct mechanisms.

\section{Experiments}\label{sec:experiments}

\begin{figure*}[t]
  \centering
  \includegraphics[width=\linewidth]{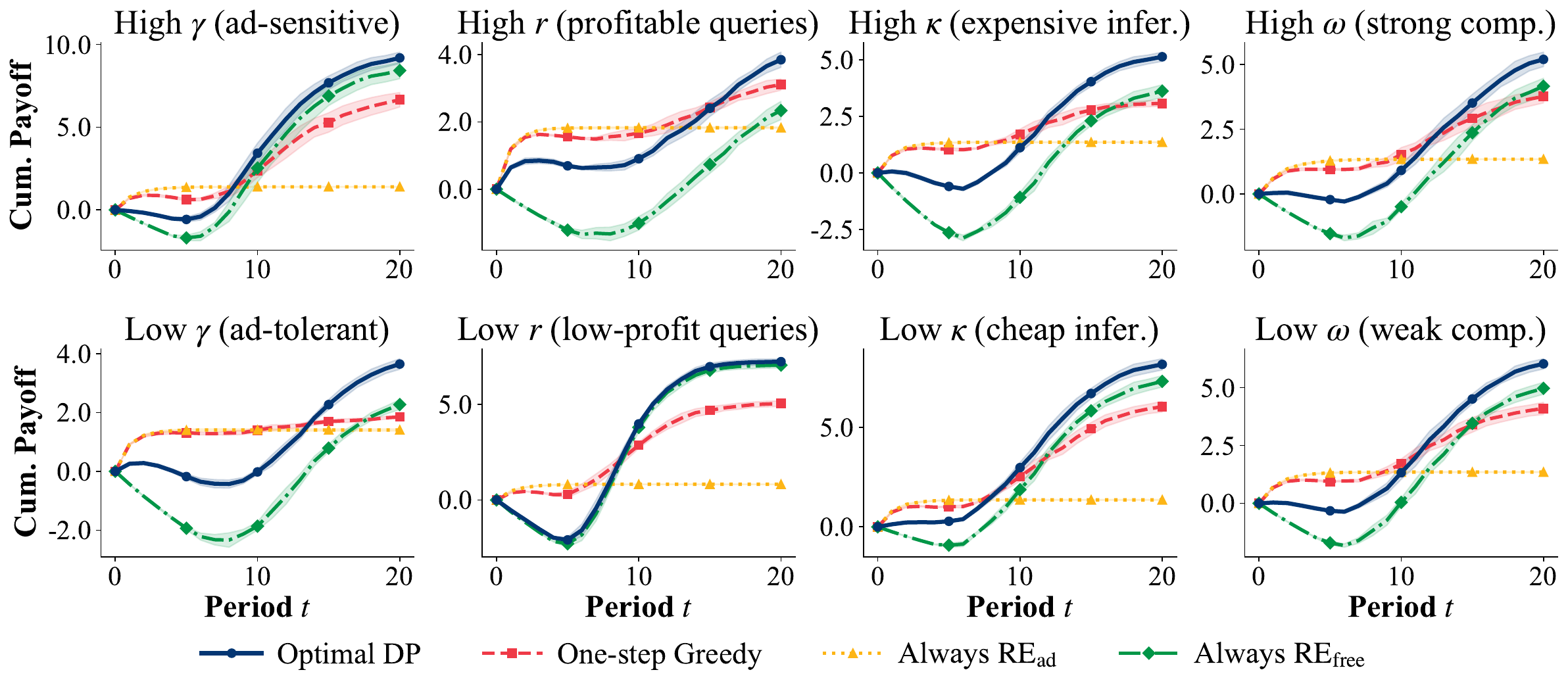}
  \vspace{-15pt}
  \caption{Cumulative payoff trajectories under eight market conditions.
  Columns correspond to the four market dimensions ($\gamma$, $r$, $\kappa$, $\omega$); rows compare high and low levels.
  Shaded bands show $\pm 1\sigma$ across five random seeds ($N{=}300$ each).}
  \label{fig:sensitivity_panel}
  \vspace{-10pt}
\end{figure*}

We conduct a simulation study to illustrate the dynamic trade-off between short-term monetization and long-term user development in design policy.
The analysis traces how different allocations between GRs with ads and ad-free GRs shape provider payoff, retention, subscription conversion, and AI exposure over time.

\subsection{Experimental Setup}\label{sec:exp_setup}

We simulate a GEs market with $N=500$ heterogeneous users over $T=20$ periods and discount factor $\beta=0.95$.
In each period, every active user issues a query; the generative engine chooses \AIa\ or \AIf, and the user probabilistically decides whether to engage, subscribe, and remain active in the next period.
Per-query engagement yields an immediate payoff and updates two states -- ad exposure $c$ and AI experience $s$ -- which govern future retention and subscription conversion.
We follow the model specification and calibration in \Cref{sec:setup,sec:dp-results}; full details are in \Cref{app:sim_details}.

\textbf{Monetization Design Policies.}
We compare four design policies deployed by the generative engine:
\begin{enumerate}[noitemsep, topsep=2pt, parsep=0pt, partopsep=0pt]
    \item \textit{Optimal DP}: computed from the Bellman equation in \Cref{sec:setup,sec:dp-results}.
    
    \item \textit{One-step greedy}: maximizes immediate expected payoff plus next-period value.
    
    \item \textit{Always \AIa}, which always displays GRs with ads.
    
    \item \textit{Always \AIf}, which always displays ad-free GRs.

\end{enumerate}

\subsection{Results}\label{sec:exp_results}

The simulations highlight three implications of the dynamic design problem: why myopic monetization underperforms, what mechanism drives the gap, and how the optimal policy balances the trade-off.

\textbf{Cumulative payoff.}
\Cref{fig:experiment_time} (panel 1) highlights the payoff trade-off: policies that over-monetize early leave long-term value on the table.
Policies that shift weight toward \AIa\ front-load revenue but weaken continuation value, whereas those that shift toward \AIf\ slow early payoff but support stronger downstream gains through retention and conversion.
At the extremes, \textit{Always \AIa} concentrates early revenue with limited continuation gains, while \textit{Always \AIf} bears the largest short-term cost but produces the strongest retention and subscription outcomes.
The \textit{Optimal DP} balances these forces by internalizing both the current-query payoff and the effect of today's design choice on future states.
By contrast, \textit{One-step greedy} remains too sensitive to short-term payoffs and therefore underweights these continuation benefits.

\textbf{Active users and subscriptions.}
\Cref{fig:experiment_time} (panels 2--3) reveal the mechanism behind the cumulative payoffs.
Policies that allocate more \AIf\ build AI experience $s$ faster.
This lifts both retention and subscription rates.
Policies that lean on \AIa\ slow that accumulation process and shrink the population of future active users.
This is why short-sighted over-monetization underperforms: it sacrifices the user base that generates future revenue.
At one extreme, \textit{Always \AIa} slows upgrading and accelerates churn; at the other, \textit{Always \AIf} maximizes conversion at the cost of short-term revenue.
The dynamic policies lie between these extremes by adjusting \AIf\ exposure over time rather than committing to a fixed display rule.

\textbf{AI exposure rate.}
We report the \AIf\ exposure rate, defined as the share of active non-subscribers shown \AIf\ by the policy.
\Cref{fig:experiment_time} (panel 4) captures how aggressively each policy allocates \AIf\ to active non-subscribers over time.
Sustained \AIf\ allocation supports retention and conversion over time (e.g., \textit{Optimal DP}, \textit{Always \AIf}), whereas short-term \AIf\ pushes do not sustain these gains to the same extent (e.g., \textit{One-step greedy}).
This contrast shows that sustained allocation of ad-free responses matters more than temporary spikes in \AIf\ exposure.

\textbf{Strategic implications for GEs.}
Together, these results show that monetization design is not a static choice but a state-dependent allocation rule that balances short-term margins against long-term engagement returns.
In practice, the policy depends on the economics of ad-free investment, the generative engine's weight on future outcomes, and the surrounding market environment.
Inference cost governs how expensive it is to sustain ad-free exposure, forward-looking orientation determines how much the generative engine values continuation gains, and market conditions shape the return to building retention and subscription demand (see \Cref{sec:dp-results}) \citep{DastinNellis2023AIBingBardSearchCosts,CaiSophia2025AlphabetCapex}.

\subsection{Sensitivity to Market Conditions}\label{sec:exp_ablation}

We test the robustness of the baseline findings through a paired sensitivity analysis.
For each of four market dimensions---ad sensitivity~($\gamma$), query profitability~($r$), inference cost~($\kappa$), and outside-option strength~($\omega$)---we simulate a ``high'' and ``low'' variant while holding all others at baseline values, yielding eight market conditions evaluated under all four policies.
Full specifications and numerical summaries are in \Cref{app:ablation_details}. \Cref{fig:sensitivity_panel} shows the cumulative payoff trajectories under the eight conditions, while Appendix Table~\Cref{tab:ablation_full} reports the detailed results.

Two economic patterns emerge from the sensitivity analysis.

\textbf{Pattern 1: The value of dynamic adaptation is robust across market conditions.}
Across all eight conditions, allowing the generative engine to adjust its display policy to the market environment yields cumulative payoff that is at least as high as under any fixed rule.
The difference is largest when a static rule is especially misaligned with the environment: Always~\AIa\ misses subscription-building gains, whereas Always~\AIf\ performs worse when free-tier inference is costly or when ad insertion is especially profitable.

\textbf{Pattern 2: Market conditions reshape the return to ad-free investment.}
Each market dimension shifts the short-term cost vs.\ long-term benefit trade-off.
High ad sensitivity strengthens the subscription channel, high inference cost raises the cost of ad-free display, and stronger outside competition compresses overall payoffs while making retention harder.
Query profitability is particularly informative: total payoff can be higher under low~$r$ than under high~$r$, because lower ad revenue reduces the opportunity cost of ad-free exposure and makes subscription-building relatively more attractive.
Across conditions, the dynamic policy adjusts its \AIf\ exposure accordingly, consistent with the comparative statics in \Cref{prop:type_order}: it shifts toward \AIf\ for ad-sensitive users and low-profit queries, and toward \AIa\ when immediate ad revenue is high.

Detailed trajectory plots for all four metrics (cumulative payoff, active users, subscribers, and AI exposure rate) under each condition are provided in \Cref{app:ablation_details}.

\section{Conclusion}
We develop a dynamic game-theoretic framework for GEs' monetization problems in which a generative engine chooses between advertising and subscription, accounting for how current choices shape future user engagement and subscription adoption.
Our analysis highlights the core trade-off: ad-free responses build AI experience, which supports retention and subscription conversion over time, while ad-heavy policies raise short-term revenue at the cost of weaker long-term value.
The optimal design shifts toward ad-free responses when the generative engine is more forward-looking, when subscription conversion becomes more valuable, and when outside competition weakens engagement-dependent ad revenue.
The appendix extends the analysis to a social-welfare objective and to learning the model parameters from logged data; see \Cref{sec:dp-welfare,sec:dp-learning}.
Future work could extend the monopoly setting to a multi-player game among competing generative engines and incorporate additional operational constraints such as capacity or latency.
As GEs scale, these results clarify when ad-free responses create long-term value and when ad-heavy policies mainly trade future retention and conversion for short-term revenue.



\bibliography{example_paper}
\bibliographystyle{icml2025}

\newpage
\appendix
\onecolumn

This appendix collects all extensions, regularity conditions, proofs, and supplementary experiments.
\Cref{app:framework_notes} extends the framework setup (probabilistic adoption, post-subscription churn, and query-dependent costs).
\Cref{app:theory_notes} states regularity conditions, collects all proofs, and develops two extensions: a social-welfare objective (\Cref{sec:dp-welfare}) and learning DP primitives from data (\Cref{sec:dp-learning}).
\Cref{app:experiment_notes} provides additional experimental details.

\section{Notes for Framework Setups}\label[appsec]{app:framework_notes}

\subsection{Problem Setup and Preliminaries}

\begin{remark}[Beyond i.i.d.\ queries]
We assume $q_t \sim \mathcal P(\cdot\mid x)$ i.i.d.\ only to keep the pre-subscription state minimal.
All subsequent policy comparisons are pointwise in the realized query $q$.
More generally, if queries follow any stationary exogenous process -- for example, $q_{t+1}\sim \mathcal P(\cdot\mid x,q_t)$ -- the DP remains valid after augmenting the state with the relevant query-process variables (e.g., the current $q_t$ or a latent topic state), with no change to the per-query policy differential.
\end{remark}

\subsubsection{Probabilistic Adoption Hazard}\label[appsec]{app:prob_adoption}

This appendix relaxes the deterministic subscription threshold. After the post-query state update to $(s^+,c^+)$ under action $a\in\{\AIa,\AIf\}$, the user adopts the subscription according to a Bernoulli hazard
\begin{equation}
    Z_{t+1} \,\sim\, \mathrm{Bernoulli}\big(\pi(x,p,s^+,c^+)\big),
\end{equation}
where $\pi(x,p,s,c)$ is weakly increasing in $s$, weakly decreasing in $c$, and takes values in $[0,1]$.
The deterministic threshold rule is nested as the special case $\pi(x,p,s,c)=\mathbf 1\{s\ge \tau(x,p,c)\}$.

\textbf{Adoption-weighted continuation value.}
Define the adoption-weighted continuation value
\begin{equation}
    \widetilde V_x(s,c)
    \;:=\;
    \pi(x,p,s)\,V_x(s,c,1) \;+\; \big(1-\pi(x,p,s)\big)\,V_x(s,c,0).
\end{equation}
Then the pre-subscription Bellman operator becomes
\begin{equation}
\label{eq:bellman_prob_adopt}
    V_x(s,c,0)
    \;=\;
    \mathbb E\bigg[
        \max_{a\in\{\AIa,\AIf\}}
        \Big\{
            r_a(x,q)
            + \beta\,\rho_{z_a^+(q)}(x,s_a^+(q),c_a^+(q))\,
            \widetilde V_x\big(s_a^+(q),c_a^+(q)\big)
        \Big\}
    \bigg],
\end{equation}
where the expectation is taken over $q\sim \mathcal P(\cdot\mid x)$ (and any within-query randomness, if present).

\textbf{Implication for structural results.}
Equation~\eqref{eq:bellman_prob_adopt} shows that all subsequent policy comparisons continue to take the same ``flow vs.\ discounted continuation'' form; relative to the deterministic threshold, the only change is the replacement of $\mathbf 1\{\cdot\}$ by $\pi(\cdot)$ inside $\widetilde V_x$.
In particular, under the same regularity conditions used in the main text (monotonicity of $V_x(\cdot,\cdot,z)$ and $\rho_{z_a^+(q)}$ in $(s,c)$), the single-crossing / threshold characterizations of the optimal design policy extend directly.

\subsubsection{Experience-Dependent Churn After Subscription}\label[appsec]{app:churn}

We relax the absorbing-subscription benchmark by allowing churn that depends on the experience state.
When $Z_t=1$, the user remains subscribed into period $t+1$ with probability
\begin{equation}
    \rho_1(x,s,c)\in(0,1],
\end{equation}
and otherwise churns back to the pre-subscription regime ($Z_{t+1}=0$) with probability $1-\rho_1(x,s,c)$.
We assume $\rho_1(x,s,c)$ is weakly increasing in $s$ and weakly decreasing in $c$.

\textbf{Subscribed-state Bellman equation.}
Let the platform serve the subscribed experience (e.g., $\AIf^{\text{sub}}$) when $Z_t=1$, yielding flow payoff $r_{\text{sub}}(q,x)$ and post-query updates $(s^+_{\text{sub}}(q),c^+_{\text{sub}}(q))$.
Then the subscribed-state value satisfies
\begin{equation}
\label{eq:bellman_churn}
    V_x(s,c,1)
    \;=\;
    \mathbb E\bigg[
        r_{\text{sub}}(q,x)
        + \beta\Big(
            \rho_1\!\big(x,s^+,c^+\big)\,V_x(s^+,c^+,1)
            + \big(1-\rho_1\!\big(x,s^+,c^+\big)\big)\,V_x(s^+,c^+,0)
        \Big)
    \bigg],
\end{equation}
where $(s^+,c^+)=(s^+_{\text{sub}}(q),c^+_{\text{sub}}(q))$ and the expectation is over $q\sim\mathcal P(\cdot\mid x)$.

\textbf{Pre-subscription Bellman equation.}
The pre-subscription Bellman equation remains identical to the main text after replacing the continuation term $V_x(\cdot,\cdot,1)$ by the churn-adjusted subscribed value defined in \eqref{eq:bellman_churn}.

\begin{lemma}[Monotonicity of subscribed value under churn]\label{lem:churn_monotone}
Suppose $\rho_1(x,s,c)$ is weakly increasing in $s$ and weakly decreasing in $c$.
If $r_{\text{sub}}(q,x)$ and $(s^+_{\text{sub}}(q),c^+_{\text{sub}}(q))$ do not reverse these orders (e.g., $s^+_{\text{sub}}$ is nondecreasing in $s$ and $c^+_{\text{sub}}$ is nondecreasing in $c$), then $V_x(s,c,1)$ is weakly increasing in $s$ and weakly decreasing in $c$.
\end{lemma}

\textbf{Implication for the optimal design policy.}
Lemma~\ref{lem:churn_monotone} ensures that the continuation-value channel that drives the optimal design policy (via the prospect of reaching and remaining in $Z=1$) preserves the same comparative statics as in the no-churn benchmark.
Consequently, all policy characterizations in the main text (e.g., single-crossing / threshold structures in $s$ conditional on $(q,c)$) extend under the same sufficient conditions, with the only modification that the subscribed continuation value now solves \eqref{eq:bellman_churn}.

\subsection{Context- and Query-Dependent Inference Costs}
\label{app:kappa_xq}

We now allow the inference cost to depend on user context and query.
Fix a measurable function $\kappa:\mathcal{X}\times\mathcal{Q}\to\mathbb{R}_+$ and replace the constant cost in the pre-subscription flow payoffs by $-\kappa(x,q)$.
All other primitives (choice probabilities, state updates, retention, and conversion) are unchanged.

\begin{assumption}[Contextual inference cost enters symmetrically]\label{assm:kappa_xq}
In the pre-subscription regime ($z=0$), the per-period payoffs take the form
\[
r_{\AIa}^{\kappa}(x,q,s,c)=\bar r_{\AIa}(x,q,s,c)-\kappa(x,q),
\qquad
r_{\AIf}^{\kappa}(x,q,s,c)=\bar r_{\AIf}(x,q,s,c)-\kappa(x,q),
\]
for some baseline payoffs $(\bar r_{\AIa},\bar r_{\AIf})$ that do not depend on $\kappa$.
All transition primitives (including the within-query choice model and cross-period state updates) are independent of $\kappa$.
\end{assumption}

\begin{remark}[Direct versus continuation effects of $\kappa(x,q)$]\label{rem:kappa_xq_direct}
Under \Cref{assm:kappa_xq}, the common current-period cost term $-\kappa(x,q)$ enters both action values symmetrically, so it cancels in the immediate payoff comparison. Any effect of $\kappa(x,q)$ on the optimal design policy must therefore operate through continuation values, because future pre-subscription periods also incur the cost and the two actions can lead to different future states.
\end{remark}

\begin{proposition}[Value monotonicity under pointwise cost increases]\label{prop:kappa_xq_value}
Under \Cref{assm:kappa_xq} and \Cref{app:regularity}, let $\kappa_1,\kappa_2$ be two cost functions such that
$\kappa_2(x,q)\ge \kappa_1(x,q)$ for all $(x,q)$.
Let $V_x^{\kappa_i}$ denote the optimal value function under $\kappa_i$.
Then
\[
V_x^{\kappa_2}(s,c,z)\le V_x^{\kappa_1}(s,c,z)
\qquad \text{for all }(s,c,z).
\]
\end{proposition}

\begin{proof}
Let $\mathcal{T}_x^{\kappa}$ be the Bellman operator under cost function $\kappa$.
Under \Cref{assm:kappa_xq}, for any bounded $V$ and any pre-subscription state $(s,c,0)$,
each feasible action-value is shifted down by $\kappa(x,q)$ for the realized query $q$.
Hence the pointwise maximand is shifted down by the same amount, and after integrating over $q$ we have
\[
\mathcal{T}_x^{\kappa_2}V(s,c,0)\le \mathcal{T}_x^{\kappa_1}V(s,c,0),
\]
with equality for subscribed states $z=1$ (where no free-tier inference cost is incurred).
By the contraction property from \Cref{app:regularity}, $\mathcal{T}_x^{\kappa}$ has a unique fixed point $V_x^{\kappa}$, and monotonicity of the operator implies the fixed points satisfy
$V_x^{\kappa_2}\le V_x^{\kappa_1}$ pointwise.
\end{proof}

\begin{remark}[Uniform cost shifts]\label{rem:kappa_xq_uniform}
A common comparative-statics experiment corresponds to the one-parameter family
$\kappa_{\Delta}(x,q):=\kappa(x,q)+\Delta$.
\Cref{prop:kappa_xq_value} implies $V_x^{\kappa_{\Delta}}$ is weakly decreasing in $\Delta$ pointwise.
Whether such uniform shifts also move the optimal design policy depends on how the two actions change the expected time spent in the pre-subscription regime.
\end{remark}

\section{Notes for Theoretical Results}\label{app:theory_notes}
\subsection{Standing Regularity Conditions and Assumptions}\label[appsec]{app:regularity}

\begin{assumption}[Standing regularity conditions]\label{assm:regularity}
Fix a user context $x$ and price $p$, and let $\beta\in(0,1)$.
Throughout, queries satisfy $q\sim \mathcal P(\cdot\mid x)$ on a measurable space $(\mathcal Q,\mathcal F)$.

\begin{enumerate}
    \item \textbf{Measurability.}
    All primitives are measurable with respect to the underlying Borel $\sigma$-algebras: 
    $r_{a}(x,q)$ is measurable in $q$, $\delta_S(x,q,s)$ and $\delta_C(x,q,c)$ are measurable in $(q,s)$ and $(q,c)$ , 
    and $\rho_{x}(s,c,z)$ is measurable in $(s,c)$ for $z \in \{0,1\}$. Consequently, for each action $a\in\{\AIa,\AIf\}$, the post-decision map $(s,c,q)\mapsto (S^+_{a,Y},C^+_{a,Y},Z^+_{a,Y})$ is measurable. $\Pr(Y_t=1\mid x,q,a)$ is measurable.

    \item \textbf{Bounded one-period payoffs.}
    There exists $M<\infty$ such that
    \[
    \sup_{q\in\mathcal Q,\,c\ge 0}\big|r_{\AIa}(x,q)\big|\le M,
    \qquad 0\le \kappa<\infty,
    \qquad |r_{\sub}(p)|<\infty,
    \]
    and we set $r_{\AIf}(x,q)\equiv -\kappa$.

    \item \textbf{Well-defined retention.}
    The pre-subscription retention probability satisfies
    \[
    \rho_{x}(s,c,z)\in[0,1]\qquad\text{for all } s\ge 0,\ c\ge 0.
    \]

    \item \textbf{Nonnegative state increments.}
    The experience increments are nonnegative and finite:
    \[
    \delta_S(x,q,s)\ge 0,\qquad \delta_C(x,q,c)\ge 0 \qquad \text{for all } q\in\mathcal Q.
    \]

    \item \textbf{Immediate adoption convention (if used).}
    The subscription status at the start of a period is deterministic:
    $z=1$ if and only if $s\ge \tau(x,p,c)$, and no exit once subscribed.  Hence, in the pre-subscription regime $z=0$ we restrict attention to states with $s<\tau(x,p,c)$.
\end{enumerate}
\end{assumption}

\begin{assumption}[Type sufficiency and monotone type effects]\label{assm:type_mono}
Fix a pre-subscription state $(s,c,0)$. 
Let $Q_x^{a}(s,c,q)$ be the action value for
$a\in\{\AIa,\AIf\}$ as defined in Proposition~\ref{prop:opt_gate}. Define the AI advantage
\[
D_x(s,c,q):=Q_x^{\AIf}(s,c,q)-Q_x^{\AIa}(s,c,q).
\]
Note that $D_x=-\Delta_x$, where $\Delta_x$ is the \AIa's edge defined in \Cref{prop:opt_gate}.

Assume there exists a function $\widetilde D$ such that for all $(x,q)$,
\[
D_x(s,c,q)=\widetilde D\!\big(s,c;\gamma(x),r_{\AIa}(x,q),\psi(x,q)\big).
\]
Moreover, for each fixed $(s,c)$, the map $(\gamma,r,\psi)\mapsto \widetilde D(s,c;\gamma,r,\psi)$
is weakly increasing in $\gamma$ and $\psi$, and weakly decreasing in $r$.
Moreover, for each fixed $(s,c,\gamma,r)$, the function
$\psi\mapsto \widetilde D(s,c;\gamma,r,\psi)$ is left-continuous in $\psi$,
and for each fixed $(s,c,\gamma,\psi)$, the function
$r\mapsto \widetilde D(s,c;\gamma,r,\psi)$ is right-continuous in $r$.

When this assumption is invoked jointly with \Cref{assm:outside_comp}, we require that the above conditions hold for the $\omega$-indexed action values $Q_x^a(s,c,q;\omega)$ and the corresponding edge $D_x(s,c,q;\omega)=\widetilde D^\omega(s,c;\gamma,r,\psi)$ at each fixed $\omega\ge 0$.
\end{assumption}

\begin{assumption}[Horizon monotonicity and continuation dominance]
\label{assm:horizon_beta}
Fix $(x,p)$ and consider the family of problems indexed by $\beta\in(0,1)$.
Let $V_x(\cdot;\beta)$ denote the unique bounded value function solving the Bellman equation at discount
factor $\beta$, and let $Q_x^{a}(s,c,q;\beta)$ be the associated action value for
$a\in\{\AIa,\AIf\}$.

Define the traffic-weighted continuation value under action $a$ by
\[
W_x^{a}(s,c,q;\beta)
:=
\mathbb E_{Y\sim \Pr(\cdot\mid x,q,a)}\!\Big[
\rho_{Z_{a,Y}^+}\!\big(x,S_{a,Y}^+,C_{a,Y}^+\big)\,
V_x\!\big(S_{a,Y}^+,C_{a,Y}^+,Z_{a,Y}^+;\beta\big)
\Big],
\]
and define the continuation advantage of $\AIf$ over $\AIa$ as
\[
\Phi_x(s,c,q;\beta):=W_x^{\AIf}(s,c,q;\beta)-W_x^{\AIa}(s,c,q;\beta).
\]
Assume:
\begin{enumerate}
\item \textbf{Uniform continuation dominance of $\AIf$.}
For all $(s,c,q)$ and all $\beta\in(0,1)$,
\[
\Phi_x(s,c,q;\beta)\ge 0.
\]
\item \textbf{Continuation advantage monotonicity in $\beta$.}
For any $\beta'>\beta$ and all $(s,c,q)$,
\[
\Phi_x(s,c,q;\beta')\ge \Phi_x(s,c,q;\beta).
\]
\end{enumerate}
\end{assumption}

\begin{assumption}[Common inference cost and subscription-tier margin]\label{assm:cs-cost}
In the pre-subscription regime ($z=0$), both display options generate an AI answer and incur the same
per-query inference cost $\kappa\ge 0$. Thus the flow payoffs are
\[
r_{\AIa}(x,q)-\kappa \quad\text{under }\AIa,
\qquad
-\kappa \quad\text{under }\AIf,
\]
where $r_{\AIa}(x,q)$ depends only on $(x,q)$.
After subscription ($z=1$), subscription is absorbing and the provider earns per-period net payoff
\[
r_{\sub}(p):=p-\kappa_{\sub},
\]
with $\kappa_{\sub}\ge \kappa\ge 0$.
All other primitives $\big(\rho,\delta_S,\delta_C,\tau\big)$ are independent of $(\kappa,\kappa_{\sub})$.
We break ties in the design policy in favor of $\AIa$.
\end{assumption}

\begin{assumption}[One-step conversion region]\label{assm:one_step}
Fix $(x,p)$ and a pre-subscription decision point $(s,c,q)$ with $s<\tau(x,p,c)$.
Let $u:=\delta_S(x,q,s)$.
Assume:
\begin{enumerate}[noitemsep, topsep=0pt]
\item \textbf{One-step conversion under $\AIf$.} $s+u\ge \tau(x,p,c)$.
\item \textbf{No one-step conversion under $\AIa$.} Choosing $\AIa$ does not increase $s$ in one step, so the next-period subscription status remains $z=0$.
\item \textbf{Engagement under $\AIf$ at this decision point.}
$\Pr(Y=1\mid x,q,\AIf)=1$.

\end{enumerate}
\end{assumption}

\begin{assumption}[Price affects subscription margin and (possibly) the conversion threshold]\label{assm:cs-p-nuanced}
Fix $x$ and hold $(\kappa,\kappa_{\sub})$ fixed.
The subscription margin $r_{\sub}(p)=p-\kappa_{\sub}$ is weakly increasing in $p$.
The conversion threshold $\tau(x,p,c)$ is weakly nondecreasing in $p$ and weakly decreasing in $c$.
All pre-subscription primitives $\big(r_{\AIa},\rho,\delta_S,\delta_C\big)$ are independent of $p$.
\end{assumption}

\begin{assumption}[Outside-competition comparative statics]\label{assm:outside_comp}
Fix $(x,p)$ and index outside competition by $\omega\ge 0$.
The within-query outside-option utility is shifted by
\[
v_0(x,q;\omega)=v_0(x,q)+\omega,
\]
while $v_a(x,q)$ for $a\in\{\AIa,\AIf\}$ and all cross-period primitives
$\big(\rho,\delta_S,\delta_C,\tau\big)$ are independent of $\omega$.

Let $V_x^\omega$ denote the optimal value function under outside-option intensity $\omega$, and define
\[
W_x^\omega(s,c,z):=\rho_x(s,c,z)V_x^\omega(s,c,z).
\]

For a fixed decision point $(s,c,q)$, define
\[
p_a(\omega):=\Pr(Y=1\mid x,q,a;\omega),
\qquad a\in\{\AIa,\AIf\},
\]
and the continuation-value advantage of \AIf over \AIa by
\begin{align}
\Phi_x(s,c,q;\omega)
&:=
\mathbb E_{Y\sim \Pr(\cdot\mid x,q,\AIf;\omega)}
\Big[
W_x^\omega\big(S^+_{\AIf,Y},C^+_{\AIf,Y},Z^+_{\AIf,Y}\big)
\Big]
\nonumber
\\
&\quad -
\mathbb E_{Y\sim \Pr(\cdot\mid x,q,\AIa;\omega)}
\Big[
W_x^\omega\big(S^+_{\AIa,Y},C^+_{\AIa,Y},Z^+_{\AIa,Y}\big)
\Big].
\label{eq:outside_comp_phi}
\end{align}

Assume:
\begin{enumerate}[noitemsep, topsep=0pt]
\item \textbf{Nonnegative ad revenue.}
$R(x,q)\ge 0$.

\item \textbf{Continuation advantage strengthens with competition.}
For the fixed decision point $(s,c,q)$, the function
\[
\omega\mapsto \Phi_x(s,c,q;\omega)
\]
is weakly increasing.
\end{enumerate}
\end{assumption}

\subsection{Extra Theoretical Results}

\subsection{Proofs}
\label[appsec]{app:proof}

\begin{refproof}[Proof of Proposition~\ref{prop:opt_gate}]
\label{proof:opt_gate}
Fix $(x,p)$ and suppress these arguments when notationally convenient.

Let the state space be
\[
\mathcal S := [0,\infty)\times[0,\infty)\times\{0,1\}
\]
equipped with the product Borel $\sigma$-algebra. Let $\mathcal B$ denote the space of bounded
\emph{measurable} functions $V:\mathcal S\to\mathbb R$, equipped with the sup norm
$\|V\|_\infty:=\sup_{(s,c,z)\in\mathcal S}|V(s,c,z)|$.

\textbf{Step 1: Bellman operator.}
First let $Y\sim \Pr(\cdot\mid x,q,a)$ denote the within-query choice, and write
$(S^+,C^+,Z^+):=(S_{a,Y}^+(x,q,s),\,C_{a,Y}^+(x,q,c),\,Z_{a,Y}^+(x,q,s))$.
Define the Bellman operator $\mathcal T_x:\mathcal B\to\mathcal B$ by
\begin{align}
(\mathcal T_x V)(s,c,0)
&:= \int_{\mathcal Q} \max_{a\in\{\AIa,\AIf\}}
\Big\{
r_{a}(x,q)
+\beta \,
\mathbb E_{Y\sim \Pr(\cdot\mid x,q,a)}
\big[  
\rho\,\!\big(S^+,C^+,Z^+\big)\,
V\!\big(S^+,C^+,Z^+\big)
\big]
\Big\}\,
\mathcal P(dq\mid x),
\label{eq:bellman_z0}
\\
(\mathcal T_x V)(s,c,1)
&:= r_{\sub}(p)+\beta\,V(s,c,1).
\label{eq:bellman_z1}
\end{align}
Under Assumption~\ref{assm:regularity}, the integrand in \eqref{eq:bellman_z0} is measurable in $q$:
$r_{\AIa}(x,q)$ is measurable in $(q,c)$, $\delta_S(x,q,s)$ and $\delta_C(x,q,c)$ are measurable in $q$,
and $V(\cdot)$ is measurable in $(s,c,z)$; hence the composition
$q\mapsto V(s_a^+(q),c_a^+(q),z_a^+(q))$ is measurable. Since the max of finitely many measurable functions
is measurable, the integral is well-defined. Moreover, boundedness of one-period payoffs and $V$ implies
$\mathcal T_x V$ is bounded.

\textbf{Step 2: Contraction.}
Let $V,W\in\mathcal B$. For any $(s,c)$ and $q$, define
\[
G_V(a;q)
:=
r_a(x,q)
+
\beta\,
\mathbb E_{Y\sim \Pr(\cdot\mid x,q,a)}
\big[  
\rho\,\!\big(S^+,C^+,Z^+\big)\,
V\!\big(S^+,C^+,Z^+\big)
\big]
\]
Then for each $(s,c)$,
\begin{align*}
|(\mathcal T_x V)(s,c,0)-(\mathcal T_x W)(s,c,0)|
&= \left|\int \Big(\max_a G_V(a;q)-\max_a G_W(a;q)\Big)\,\mathcal P(dq\mid x)\right| \\
&\le \int \max_a |G_V(a;q)-G_W(a;q)|\,\mathcal P(dq\mid x) \\
&\le \beta\,\mathbb E_{Y\sim \Pr(\cdot\mid x,q,a)}\big[\rho(\cdot)\,\|V-W\|_\infty\, \big]
\;\le\; \beta\,\|V-W\|_\infty,
\end{align*}
where we used that $\rho_{z_a^+(q)}\in[0,1]$.
For $z=1$, \eqref{eq:bellman_z1} gives
\[
|(\mathcal T_x V)(s,c,1)-(\mathcal T_x W)(s,c,1)|
= \beta\,|V(s,c,1)-W(s,c,1)|
\le \beta\,\|V-W\|_\infty.
\]
Taking the supremum over $(s,c,z)$ yields
\[
\|\mathcal T_x V-\mathcal T_x W\|_\infty \le \beta\,\|V-W\|_\infty.
\]
Hence $\mathcal T_x$ is a contraction on $(\mathcal B,\|\cdot\|_\infty)$.

\textbf{Step 3: Existence and uniqueness of the value function.}
By the Banach fixed-point theorem, $\mathcal T_x$ admits a unique fixed point $V_x\in\mathcal B$ satisfying
$V_x=\mathcal T_x V_x$. In particular, \eqref{eq:bellman_z0} implies
\[
V_x(s,c,0)
= \int_{\mathcal Q} \max_{a\in\{\AIa,\AIf\}} Q_x^{a}(s,c,q)\,\mathcal P(dq\mid x),
\]
where $Q_x^a$ is the action-value function defined in \Cref{prop:opt_gate}. Equation \eqref{eq:bellman_z1} yields
$V_x(s,c,1)=r_{\sub}(p)/(1-\beta)$.

\textbf{Step 4: Optimality and the sign characterization.}
Consider a period in which the platform is in the pre-subscription regime with state $(s,c)$ and observes
a realized query $q$. Conditional on $(s,c,q)$, the only choice is $a\in\{\AIa,\AIf\}$, and the continuation value
under any admissible policy equals the corresponding action value $Q_x^{a}(s,c,q)$ computed using the fixed point
$V_x$. Therefore an optimal decision rule at $(s,c,q)$ selects any maximizer of $Q_x^{a}(s,c,q)$.

Since the action set is binary, define
\[
\Delta_x(s,c,q):=Q_x^{\AIa}(s,c,q)-Q_x^{\AIf}(s,c,q).
\]
Then $Q_x^{\AIa}(s,c,q)\ge Q_x^{\AIf}(s,c,q)$ if and only if $\Delta_x(s,c,q)\ge 0$.
With ties broken toward $\AIa$, an optimal design decision is
\[
a^\star(s,c,q)=
\begin{cases}
\AIa, & \Delta_x(s,c,q)\ge 0,\\
\AIf, & \Delta_x(s,c,q)<0,
\end{cases}
\]
which proves the claim. \qedhere
\end{refproof}

\begin{proof}[Proof of \Cref{prop:type_order}]
Fix $(s,c)$ in the pre-subscription regime.
By \Cref{prop:opt_gate}, conditional on $(x,q,s,c)$, $\AIa$ is optimal
if and only if $Q_x^{\AIa}(s,c,q)\ge Q_x^{\AIf}(s,c,q)$, equivalently $D_x(s,c,q)\le 0$.

Under Assumption~\ref{assm:type_mono}, we can write
\[
D_x(s,c,q)=\widetilde D(s,c;\gamma,r,\psi),
\quad\text{where }(\gamma,r,\psi)=\big(\gamma(x),r_{\AIa}(x,q),\psi(x,q)\big).
\]
Define the type acceptance set
\[
\mathcal A(s,c):=\{(\gamma,r,\psi):\ \widetilde D(s,c;\gamma,r,\psi)\ge 0\}.
\]
This is exactly the set of types for which $\AIf$ is optimal (since $\AIf$ is optimal iff $D_x\ge 0$).

\textbf{Step 1: $\mathcal A(s,c)$ is an upper set under $\preceq$.}
Take any $t\preceq t'$, i.e., $\gamma'\ge\gamma$, $\psi'\ge\psi$, and $r'\le r$.
By Assumption~\ref{assm:type_mono}, $\widetilde D(s,c;\cdot)$ is weakly increasing in $\gamma,\psi$
and weakly decreasing in $r$, hence
\[
\widetilde D(s,c;t')\;\ge\;\widetilde D(s,c;t).
\]
Therefore, if $t\in\mathcal A(s,c)$ then $\widetilde D(s,c;t)\ge 0$, implying $\widetilde D(s,c;t')\ge 0$,
so $t'\in\mathcal A(s,c)$. This proves the monotone ordering claim.

\textbf{Step 2: threshold in $\psi$ holding $(\gamma,r)$ fixed.}
Fix $(\gamma,r)$ and define $f(\psi):=\widetilde D(s,c;\gamma,r,\psi)$.
By Assumption~\ref{assm:type_mono}, $f$ is weakly increasing.
Let $S:=\{\psi\in\mathbb R:\ f(\psi)\le 0\}$. Then $S$ is a (possibly empty) lower interval.
Define $\psi^*(s,c;\gamma,r):=\sup S$ with the convention $\sup\emptyset=-\infty$.
Then $S=(-\infty,\psi^*(s,c;\gamma,r)]$, hence
\[
\AIa \text{ is optimal } \Longleftrightarrow \psi \le \psi^*(s,c;\gamma,r),
\]
which proves the stated threshold characterization in $\psi$.

\textbf{Step 3: threshold in $r$ holding $(\gamma,\psi)$ fixed.}
Fix $(\gamma,\psi)$ and define $g(r):=\widetilde D(s,c;\gamma,r,\psi)$.
By Assumption~\ref{assm:type_mono}, $g$ is weakly decreasing.
Let $R:=\{r\in\mathbb R:\ g(r)\le 0\}$. Then $R$ is a (possibly empty) upper interval.
Define $r^*(s,c;\gamma,\psi):=\inf R$ with the convention $\inf\emptyset=+\infty$.
Then $R=[r^*(s,c;\gamma,\psi),\infty)$, hence
\[
\AIa \text{ is optimal } \Longleftrightarrow r \ge r^*(s,c;\gamma,\psi),
\]
which proves the stated threshold characterization in $r$.

\end{proof}

\begin{proof}[Proof of \Cref{prop:beta_monotone}]
Fix $(x,p)$ and $(s,c,q)$. For each $\beta$, expanding the action values yields
\begin{equation}
\Delta_x(s,c,q;\beta)
=
r_{\AIa}(x,q)
-
\beta\,\Phi_x(s,c,q;\beta),
\label{eq:Delta_beta_identity}
\end{equation}
where $\Phi_x$ is defined in Assumption~\ref{assm:horizon_beta}.
Let $\beta'>\beta$. Subtracting \eqref{eq:Delta_beta_identity} at $\beta$ from the same identity at $\beta'$
gives
\begin{equation}
\Delta_x(s,c,q;\beta')-\Delta_x(s,c,q;\beta)
=
-\Big(\beta'\Phi_x(s,c,q;\beta')-\beta\Phi_x(s,c,q;\beta)\Big).
\label{eq:Delta_beta_diff}
\end{equation}
Rewrite the bracketed term as
\[
\beta'\Phi(\beta')-\beta\Phi(\beta)
=(\beta'-\beta)\Phi(\beta)+\beta'\big(\Phi(\beta')-\Phi(\beta)\big),
\]
where $\Phi(\beta)$ is shorthand for $\Phi_x(s,c,q;\beta)$.
By Assumption~\ref{assm:horizon_beta}(i), $\Phi(\beta)\ge 0$, so $(\beta'-\beta)\Phi(\beta)\ge 0$.
By Assumption~\ref{assm:horizon_beta}(ii), $\Phi(\beta')-\Phi(\beta)\ge 0$, so the second term is also nonnegative.
Hence $\beta'\Phi(\beta')-\beta\Phi(\beta)\ge 0$, and \eqref{eq:Delta_beta_diff} implies
$\Delta_x(s,c,q;\beta')\le \Delta_x(s,c,q;\beta)$, proving $\Delta$ is weakly decreasing in $\beta$.

For set inclusion, if $\Delta_x(s,c,q;\beta)<0$, then
$\Delta_x(s,c,q;\beta')\le \Delta_x(s,c,q;\beta)<0$, so $(s,c,q)$ remains in the $\AIf$ region at $\beta'$.
\end{proof}

\begin{proof}[Proof for \Cref{prop:cs-cost-local}]
Fix $(x,p,\beta)$ and suppress these arguments for convenience.

\textbf{Part (i).}
Take $\kappa_2\ge \kappa_1$ and write $\Delta\kappa:=\kappa_2-\kappa_1$.
Under \Cref{assm:one_step}, choosing \AIf at $(s,c,q)$ yields one-step conversion, so the next-period state is subscribed.
Since subscription is absorbing, for $i\in\{1,2\}$ we have
\[
Q_x^{\AIf}(s,c,q;\kappa_i,\kappa_{\sub})
=
(-\kappa_i)+\beta\,\frac{r_{\sub}(p)}{1-\beta},
\]
hence
\[
Q_x^{\AIf}(s,c,q;\kappa_2,\kappa_{\sub})-Q_x^{\AIf}(s,c,q;\kappa_1,\kappa_{\sub})
=
-\Delta\kappa.
\]

Now consider choosing \AIa at $(s,c,q)$.
By \Cref{assm:one_step}, the next-period subscription status remains $z=0$.
Let $V_{x,i}$ denote the optimal value function under $(\kappa_i,\kappa_{\sub})$.
By \Cref{prop:kappa_xq_value}, specialized to constant cost functions $\kappa_1$ and $\kappa_2$ while holding $\kappa_{\sub}$ fixed, the value from any pre-subscription state is weakly lower under $\kappa_2$:
\[
V_{x,2}(s',c',0)\le V_{x,1}(s',c',0)\qquad \text{for all }(s',c').
\]
Therefore
\begin{align*}
&Q_x^{\AIa}(s,c,q;\kappa_2,\kappa_{\sub})-Q_x^{\AIa}(s,c,q;\kappa_1,\kappa_{\sub})
\\
&=
\big(r_{\AIa}(x,q)-\kappa_2\big)-\big(r_{\AIa}(x,q)-\kappa_1\big)
\nonumber
\\
&\quad
+\beta\,\mathbb E_{Y\sim \Pr(\cdot\mid x,q,\AIa)}\!\Big[
\rho_x\!\big(S^+_{\AIa,Y},C^+_{\AIa,Y},0\big)
\nonumber
\\
&\qquad\qquad \times
\Big(
V_{x,2}\!\big(S^+_{\AIa,Y},C^+_{\AIa,Y},0\big)
-V_{x,1}\!\big(S^+_{\AIa,Y},C^+_{\AIa,Y},0\big)
\Big)
\Big]
\\
&\le
-\Delta\kappa,
\end{align*}
because $\rho_x(\cdot)\in[0,1]$ and the bracketed value difference is nonpositive.
Consequently,
\begin{align*}
\Delta_x(s,c,q;\kappa_2,\kappa_{\sub})-\Delta_x(s,c,q;\kappa_1,\kappa_{\sub})
&=
\big(Q_x^{\AIa}(\kappa_2)-Q_x^{\AIa}(\kappa_1)\big)
\nonumber
\\
&\quad-
\big(Q_x^{\AIf}(\kappa_2)-Q_x^{\AIf}(\kappa_1)\big)
\\
&\le 0.
\end{align*}
Thus $\Delta_x(s,c,q;\kappa,\kappa_{\sub})$ is weakly decreasing in $\kappa$, so a higher pre-subscription cost pushes the decision toward \AIf in the near-threshold region.

\textbf{Part (ii).}
Consider $\kappa_{\sub,2}\ge \kappa_{\sub,1}$ and write
$r_{\sub,i}:=p-\kappa_{\sub,i}$ for $i\in\{1,2\}$, so $r_{\sub,2}=r_{\sub,1}-(\kappa_{\sub,2}-\kappa_{\sub,1})$.

Under the one-step conversion condition $s<\tau(x,p,c)\le s+u$, choosing $\AIf$ moves the experience state to at least $\tau(x,p,c)$ in one step, so the next-period state is subscribed.
Since subscription is absorbing, the subscribed-state value satisfies
$V_x(s',c',1)=r_{\sub}(p)/(1-\beta)$ for any $(s',c')$.
Therefore,
\[
Q_x^{\AIf}(s,c,q;\kappa,\kappa_{\sub,i})
=
(-\kappa)+\beta\,\frac{r_{\sub,i}}{1-\beta},
\]
and hence
\[
Q_x^{\AIf}(s,c,q;\kappa,\kappa_{\sub,2})-Q_x^{\AIf}(s,c,q;\kappa,\kappa_{\sub,1})
=
-\frac{\beta}{1-\beta}\big(\kappa_{\sub,2}-\kappa_{\sub,1}\big).
\]

Now consider choosing $\AIa$ at $(s,c,q)$ with $s<\tau(x,p,c)$ and $\AIa$ not changing $s$ in one step.
Then the next-period status remains pre-subscription ($z=0$), so subscription payoffs that depend on $\kappa_{\sub}$
cannot begin before period $t+2$. Even in the most extreme case where the user becomes subscribed from $t+2$ onward with probability one,
the total discounted exposure to the subscribed payoff stream from period $t+2$ is at most
$\sum_{k=2}^{\infty}\beta^k=\beta^2/(1-\beta)$.
Because changing $\kappa_{\sub}$ shifts the subscribed per-period payoff by exactly $-(\kappa_{\sub,2}-\kappa_{\sub,1})$,
this implies the action value under $\AIa$ satisfies the bound
\[
Q_x^{\AIa}(s,c,q;\kappa,\kappa_{\sub,2})-Q_x^{\AIa}(s,c,q;\kappa,\kappa_{\sub,1})
\ \ge\
-\frac{\beta^2}{1-\beta}\big(\kappa_{\sub,2}-\kappa_{\sub,1}\big).
\]

Combining the two differences yields
\begin{align*}
\Delta_x(s,c,q;\kappa,\kappa_{\sub,2})-\Delta_x(s,c,q;\kappa,\kappa_{\sub,1})
&=
\big(Q_x^{\AIa}(\kappa_{\sub,2})-Q_x^{\AIa}(\kappa_{\sub,1})\big)
-\big(Q_x^{\AIf}(\kappa_{\sub,2})-Q_x^{\AIf}(\kappa_{\sub,1})\big)
\\
&\ge
-\frac{\beta^2}{1-\beta}\Delta\kappa_{\sub}
+\frac{\beta}{1-\beta}\Delta\kappa_{\sub}
=
\beta\,\Delta\kappa_{\sub}
\ \ge\ 0,
\end{align*}
where $\Delta\kappa_{\sub}:=\kappa_{\sub,2}-\kappa_{\sub,1}$.
This proves $\Delta_x$ is weakly increasing in $\kappa_{\sub}$ under the stated one-step condition, and the set inclusion follows immediately.
\end{proof}

\begin{proof}[Proof for \Cref{prop:cs-p-local}]
Take any $p_2>p_1$ in $I$ and write $\Delta p:=p_2-p_1>0$.

Under \Cref{assm:one_step}, choosing \AIf at $(x,s,c,q)$ yields one-step conversion for both prices.
Under \Cref{assm:cs-cost}, the current-period flow payoff under \AIf is $-\kappa$, and since subscription is absorbing,
\[
Q_x^{\AIf}(s,c,q;\,p_i)
=
-\kappa+\beta\,\frac{p_i-\kappa_{\sub}}{1-\beta},
\qquad i\in\{1,2\}.
\]
Hence
\[
Q_x^{\AIf}(s,c,q;\,p_2)-Q_x^{\AIf}(s,c,q;\,p_1)
=
\frac{\beta}{1-\beta}\Delta p.
\]

Now consider choosing \AIa at $(x,s,c,q)$.
By \Cref{assm:one_step}(ii), the next-period subscription status remains $z=0$ for both prices.
By \Cref{assm:cs-p-nuanced}, all pre-subscription primitives are independent of $p$, so the only effect of a higher price comes through the subscribed payoff stream.
Moreover, $\tau(x,p,c)$ is weakly nondecreasing in $p$, so raising $p$ cannot make subscription occur earlier; it can only weakly delay or reduce future conversion.

Therefore, even in the most favorable case for \AIa under $p_2$, the additional gain from raising the price from $p_1$ to $p_2$ is bounded above by the discounted value of receiving the higher subscribed per-period margin $\Delta p$ from period $t+2$ onward with probability one:
\[
Q_x^{\AIa}(s,c,q;\,p_2)-Q_x^{\AIa}(s,c,q;\,p_1)
\le
\sum_{k=2}^{\infty}\beta^k\,\Delta p
=
\frac{\beta^2}{1-\beta}\Delta p.
\]

Combining the two action-value differences yields
\begin{align*}
\Delta_x(s,c,q;\,p_2)-\Delta_x(s,c,q;\,p_1)
&=
\big(Q_x^{\AIa}(p_2)-Q_x^{\AIa}(p_1)\big)
\nonumber
\\
&\quad-
\big(Q_x^{\AIf}(p_2)-Q_x^{\AIf}(p_1)\big)
\\
&\le
\frac{\beta^2}{1-\beta}\Delta p
-
\frac{\beta}{1-\beta}\Delta p
=
-\beta\,\Delta p
\le 0.
\end{align*}
Hence $\Delta_x(s,c,q;\,p)$ is weakly decreasing in $p$ on $I$, proving that a higher subscription price weakly shifts the optimal design toward \AIf in this near-threshold region.
\end{proof}

\begin{proof}[Proof for \Cref{prop:outside_competition}]
Fix $(x,p,s,c,q)$ and suppress fixed arguments in notation.
Under \Cref{assm:outside_comp}, for each action $a\in\{\AIa,\AIf\}$,
\[
p_a(\omega)=\frac{\exp(v_a)}{\exp(v_a)+\exp(v_0+\omega)},
\qquad
\frac{d}{d\omega}p_a(\omega)=-p_a(\omega)\big(1-p_a(\omega)\big).
\]
Hence $p_a(\omega)$ is weakly decreasing in $\omega$ for both actions, and in particular
$p_{\AIa}(\omega)$ is weakly decreasing in $\omega$.

For each $\omega$, define
\[
C_{\AIa}(\omega)
:=
\mathbb E_{Y\sim \Pr(\cdot\mid x,q,\AIa;\omega)}
\Big[
W_x^\omega\big(S^+_{\AIa,Y},C^+_{\AIa,Y},Z^+_{\AIa,Y}\big)
\Big],
\]
\[
C_{\AIf}(\omega)
:=
\mathbb E_{Y\sim \Pr(\cdot\mid x,q,\AIf;\omega)}
\Big[
W_x^\omega\big(S^+_{\AIf,Y},C^+_{\AIf,Y},Z^+_{\AIf,Y}\big)
\Big].
\]
Then the action values at the fixed decision point are
\[
Q_x^{\AIa}(s,c,q;\omega)
=
-\kappa + R(x,q)p_{\AIa}(\omega)+\beta C_{\AIa}(\omega),
\]
\[
Q_x^{\AIf}(s,c,q;\omega)
=
-\kappa + \beta C_{\AIf}(\omega).
\]
Therefore, using $\Delta_x=Q_x^{\AIa}-Q_x^{\AIf}$ and the definition
$\Phi_x(s,c,q;\omega)=C_{\AIf}(\omega)-C_{\AIa}(\omega)$ from \Cref{eq:outside_comp_phi},
\[
\Delta_x(s,c,q;\omega)
=
R(x,q)p_{\AIa}(\omega)
\;-\;
\beta \Phi_x(s,c,q;\omega).
\]
The first term is weakly decreasing in $\omega$ because $R(x,q)\ge 0$ and
$p_{\AIa}(\omega)$ is weakly decreasing.
The second term is also weakly decreasing because
$\Phi_x(s,c,q;\omega)$ is weakly increasing by \Cref{assm:outside_comp}(ii).
Hence $\Delta_x(s,c,q;\omega)$ is weakly decreasing in $\omega$.

Therefore, if $\Delta_x(s,c,q;\omega)<0$ at some $\omega$, then for any $\omega'>\omega$,
$\Delta_x(s,c,q;\omega')\le \Delta_x(s,c,q;\omega)<0$.
This gives the stated set inclusion for the $\AIf$-optimal region.
\end{proof}

\begin{proof}[Proof of \Cref{prop:competition_type_shift}]
Fix $(s,c)$ in the pre-subscription regime.
By \Cref{prop:outside_competition}, for every $(x,q)$ and $\omega'>\omega$,
$\Delta_x(s,c,q;\omega')\le \Delta_x(s,c,q;\omega)$.
Under \Cref{assm:type_mono} (assumed to hold for each fixed $\omega$), write
$\Delta_x(s,c,q;\omega)=\widetilde\Delta^\omega(s,c;\gamma,r,\psi)$
where $\widetilde\Delta^\omega$ is weakly decreasing in $\gamma$ and $\psi$, and weakly increasing in $r$.

\textbf{Part (i).}
Fix $(\gamma,r)$ and define $f^\omega(\psi):=\widetilde\Delta^\omega(s,c;\gamma,r,\psi)$.
Since $f^\omega$ is weakly decreasing in $\psi$ and $f^{\omega'}(\psi)\le f^\omega(\psi)$ pointwise,
the threshold $\psi^*(\omega):=\sup\{\psi: f^\omega(\psi)\ge 0\}$ satisfies $\psi^*(\omega')\le \psi^*(\omega)$.

\textbf{Part (ii).}
Fix $(\gamma,\psi)$ and define $g^\omega(r):=\widetilde\Delta^\omega(s,c;\gamma,r,\psi)$.
Since $g^\omega$ is weakly increasing in $r$ and $g^{\omega'}(r)\le g^\omega(r)$ pointwise,
the threshold $r^*(\omega):=\inf\{r: g^\omega(r)\ge 0\}$ satisfies $r^*(\omega')\ge r^*(\omega)$.

\textbf{Part (iii).}
The \AIa-optimal region at $\omega$ is $\mathcal A(\omega):=\{(\gamma,r,\psi):\widetilde\Delta^\omega\ge 0\}$.
Since $\widetilde\Delta^{\omega'}\le\widetilde\Delta^\omega$ pointwise, $\mathcal A(\omega')\subseteq\mathcal A(\omega)$.
For any measure $\mu$, $\mu(\mathcal A(\omega'))\le \mu(\mathcal A(\omega))$.
\end{proof}

\subsection{Implications of the GE Provider's Weight on Future Outcomes}\label[appsec]{app:dp-horizon-beta}

How much the generative engine values future outcomes relative to immediate value directly shapes the design policy.

In our framework as introduced in \Cref{sec:setup}, the generative engine's weight on future outcomes is summarized by the discount factor $\beta\in(0,1)$.
A higher $\beta$ corresponds to a more forward-looking generative engine that places greater weight on future retention and subscription conversion relative to near-term value \citep{bajari2007estimating}.
Accordingly, we treat the previously defined value functions (and the associated other functions) as depending on $\beta$, and study how the optimal design region changes as $\beta$ varies.

\textbf{Indexing the weight.}
We index the weight of value function $V_{x}$ and action values $Q_{x}$ by $\beta$, and write the edge function as $\Delta_x(s,c,q;\beta)$ (the same object as in \Cref{sec:dp}, now viewed as a function of $\beta$).
Because one-period payoffs (first term of \eqref{eq:delta_tradeoff}) do not depend on $\beta$, the weight on future outcomes effect affects entirely through the long-term value (second term of \eqref{eq:delta_tradeoff}).

\textbf{Directional weight effects.}
To make the weight effect unambiguous, we impose two sufficient conditions (formalized in \Cref{app:regularity}).
First, a larger weight on future outcomes increases the value of future outcomes from any state.
Second, at any decision point, showing \AIf leads to weakly better future outcomes than showing \AIa.
Together, these conditions imply that a more forward-looking generative engine selects \AIf more often.

\begin{proposition}[More forward-looking generative engines show \AIf more often]
\label{prop:beta_monotone}
Under \Cref{assm:regularity,assm:horizon_beta}, for any fixed $(s,c,q)$,
the \AIa's edge of value over \AIf $\Delta_x(s,c,q;\beta)$ is weakly decreasing in $\beta$.
Consequently, for any $\beta'>\beta$,
\begin{align}
\big\{(s,c,q): \Delta_x(s,c,q;\beta)<0\big\}
\subseteq
\nonumber
\\
\big\{(s,c,q): \Delta_x(s,c,q;\beta')<0\big\}.
\end{align}
Equivalently, the decision set of \AIf weakly expands with larger $\beta$.
\end{proposition}

\textbf{Economic interpretation.}
A larger weight on future outcomes (i.e., a larger $\beta$) increases the generative engine's weight on users' future engagement and conversion to subscription relative to immediate ad revenue.
Since \AIf improves continuation value $V_{x}$ through retention $\rho_{x}$ and the user's experience accumulation, a more forward-looking generative engine displays \AIf to a larger population of users (and queries).

\Cref{fig:type_pdf_shift} (left) illustrates \Cref{prop:beta_monotone} at a fixed state and query environment.
As $\beta$ increases, the policy cutoff in \textit{ad sensitivity} shifts left: \AIf\ becomes optimal for a larger share of users (i.e., the region to the right of the vertical cutoff expands.)

\Cref{app:horizon} discusses finite-horizon formulations and how changes in the time horizon affect the policy shift.

\begin{figure*}[t]
  \centering
  \begin{minipage}{0.54\textwidth}
    \centering
    \includegraphics[width=\linewidth]{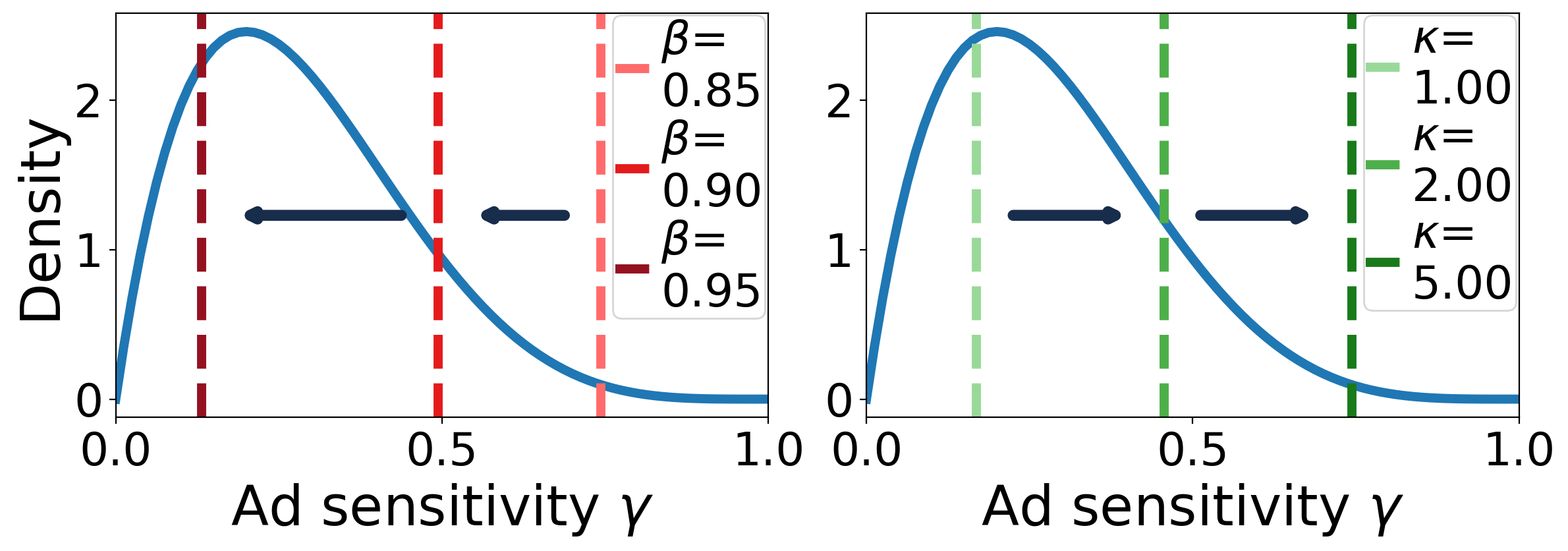}
  \end{minipage}
  \captionsetup{width=\textwidth}
  \caption{
  User-type policy cutoff shift under the generative engine's weight on future outcomes (indexed by $\beta$) and inference cost $\kappa$.
  Fix user context $x$ and query type, and draw \textit{ad sensitivity} $\gamma\sim\mathrm{Beta}$.
  Vertical dashed lines mark the optimal policy cutoff between \AIf and \AIa for each parameter value.
  \emph{Left:} varying discount factor $\beta$.
  \emph{Right:} varying AI inference cost $\kappa$.
  }
  \label{fig:type_pdf_shift}
  \vspace{-15pt}
\end{figure*}

\subsection{Finite-Horizon Case}\label{app:horizon}


\begin{proposition}[Finite-horizon value converges geometrically]
\label{prop:finite_value_convergence}
Fix $(x,p)$ and $\beta\in(0,1)$. Maintain Assumption~\ref{assm:regularity}.
Let $\mathcal T_x$ be the Bellman operator at discount factor $\beta$ (as in \Cref{proof:opt_gate}),
and define the finite-horizon value-iteration sequence by $V_x^{(0)}\equiv 0$ and
\[
V_x^{(T)} := \mathcal T_x V_x^{(T-1)} \qquad \text{for } T\ge 1.
\]
Let $V_x^\infty$ denote the unique bounded fixed point of $\mathcal T_x$.
Assume the one-period reward is uniformly bounded: there exists $R_{\max}<\infty$ such that for all
$(s,c,z)$ and all $q$,
\[
\max_{a\in\{\AIa,\AIf\}} \big| r_a(x,q) \big| \le R_{\max},
\qquad
|r_{\sub}(p)|\le R_{\max}.
\]
Then $V_x^{(T)}$ converges uniformly to $V_x^\infty$ and
\begin{equation}
\big\|V_x^\infty - V_x^{(T)}\big\|_\infty
\;\le\;
\frac{R_{\max}}{1-\beta}\,\beta^{T}.
\label{eq:finite_value_bound}
\end{equation}
\end{proposition}

\begin{proof}
Under Assumption~\ref{assm:regularity}, $\mathcal T_x$ maps bounded measurable functions to bounded
measurable functions and is a $\beta$-contraction under the sup norm:
\[
\|\mathcal T_x V - \mathcal T_x W\|_\infty \le \beta\,\|V-W\|_\infty
\qquad \text{for all bounded } V,W.
\]
Hence $\mathcal T_x$ admits a unique bounded fixed point $V_x^\infty$ and, for any initial $V_x^{(0)}$,
value iteration converges:
\[
\|V_x^\infty - V_x^{(T)}\|_\infty
\le
\beta^T \|V_x^\infty - V_x^{(0)}\|_\infty.
\]
With $V_x^{(0)}\equiv 0$, it remains to bound $\|V_x^\infty\|_\infty$.
For any bounded $V$ and any state, by the reward bound and $\rho_z\in[0,1]$,
\[
|(\mathcal T_x V)(s,c,0)|
\le
\mathbb E\!\left[\max_a |r_a(x,q)|\right] + \beta\,\|V\|_\infty
\le
R_{\max}+\beta\,\|V\|_\infty,
\]
and similarly $|(\mathcal T_x V)(s,c,1)|\le R_{\max}+\beta\,\|V\|_\infty$.
Applying this inequality to the fixed point $V_x^\infty=\mathcal T_x V_x^\infty$ yields
$\|V_x^\infty\|_\infty \le R_{\max}/(1-\beta)$.
Therefore
\[
\|V_x^\infty - V_x^{(T)}\|_\infty
\le
\beta^T \|V_x^\infty\|_\infty
\le
\frac{R_{\max}}{1-\beta}\,\beta^{T},
\]
which proves \eqref{eq:finite_value_bound}.
\end{proof}

\begin{proposition}[Policy stabilizes away from an indifference band]
\label{prop:finite_policy_stabilization}
Fix $(x,p)$ and $\beta\in(0,1)$. Maintain Assumption~\ref{assm:regularity} and the reward bound of
Proposition~\ref{prop:finite_value_convergence}.
For $a\in\{\AIa,\AIf\}$, define the infinite-horizon action value
\[
Q_{x}^{a,\infty}(s,c,q)
:=
r_a(x,q)
+
\beta\,\rho_{z_a^+(q)}\!\big(x,s_a^+(q),c_a^+(q)\big)\,
V_x^\infty\!\big(s_a^+(q),c_a^+(q),z_a^+(q)\big),
\]
and the $T$-step (finite-horizon) action value
\[
Q_{x}^{a,(T)}(s,c,q)
:=
r_a(x,q)
+
\beta\,\rho_{z_a^+(q)}\!\big(x,s_a^+(q),c_a^+(q)\big)\,
V_x^{(T-1)}\!\big(s_a^+(q),c_a^+(q),z_a^+(q)\big),
\qquad T\ge 1.
\]
Let the corresponding advantages be
\[
D_x^\infty(s,c,q):=Q_x^{\AIf,\infty}(s,c,q)-Q_x^{\AIa,\infty}(s,c,q),
\qquad
D_x^{(T)}(s,c,q):=Q_x^{\AIf,(T)}(s,c,q)-Q_x^{\AIa,(T)}(s,c,q).
\]
Then
\begin{equation}
\sup_{s,c,q}\, \big|D_x^\infty(s,c,q)-D_x^{(T)}(s,c,q)\big|
\;\le\;
2\beta\,\big\|V_x^\infty - V_x^{(T-1)}\big\|_\infty
\;\le\;
\frac{2\beta R_{\max}}{1-\beta}\,\beta^{T-1}.
\label{eq:finite_adv_bound}
\end{equation}
Consequently, if at some $(s,c,q)$,
\begin{equation}
\big|D_x^\infty(s,c,q)\big|
>
\frac{2\beta R_{\max}}{1-\beta}\,\beta^{T-1},
\label{eq:margin_condition}
\end{equation}
then $\mathrm{sign}\!\big(D_x^{(T)}(s,c,q)\big)=\mathrm{sign}\!\big(D_x^\infty(s,c,q)\big)$, and thus the
finite-horizon and infinite-horizon optimal display decisions coincide at $(s,c,q)$ (under the same tie-breaking rule).
\end{proposition}

\begin{proof}
Fix $(s,c,q)$. For any $a\in\{\AIa,\AIf\}$, the only difference between $Q_x^{a,\infty}$ and $Q_x^{a,(T)}$
is the continuation value, hence
\begin{align*}
\big|Q_x^{a,\infty}(s,c,q)-Q_x^{a,(T)}(s,c,q)\big|
&=
\beta\,\rho_{z_a^+(q)}\!\big(x,s_a^+(q),c_a^+(q)\big)\,
\big|V_x^\infty(\cdot)-V_x^{(T-1)}(\cdot)\big| \\
&\le
\beta\,\big\|V_x^\infty - V_x^{(T-1)}\big\|_\infty,
\end{align*}
since $\rho_{z_a^+(q)}\in[0,1]$. Therefore,
\begin{align*}
\big|D_x^\infty(s,c,q)-D_x^{(T)}(s,c,q)\big|
&=
\big|\big(Q_x^{\AIf,\infty}-Q_x^{\AIa,\infty}\big)-\big(Q_x^{\AIf,(T)}-Q_x^{\AIa,(T)}\big)\big| \\
&\le
\big|Q_x^{\AIf,\infty}-Q_x^{\AIf,(T)}\big|
+
\big|Q_x^{\AIa,\infty}-Q_x^{\AIa,(T)}\big| \\
&\le
2\beta\,\big\|V_x^\infty - V_x^{(T-1)}\big\|_\infty,
\end{align*}
which proves the first inequality in \eqref{eq:finite_adv_bound}.
The second inequality in \eqref{eq:finite_adv_bound} follows by combining with
Proposition~\ref{prop:finite_value_convergence}.
Finally, if \eqref{eq:margin_condition} holds, then
\[
|D_x^{(T)}-D_x^\infty|<|D_x^\infty|,
\]
which implies $D_x^{(T)}$ has the same sign as $D_x^\infty$. Under the same tie-breaking rule,
the induced optimal display decision therefore coincides at $(s,c,q)$.
\end{proof}

\subsection{A Social-Welfare Extension}
\label{sec:dp-welfare}

We extend the framework to study when the generative engine's revenue-maximizing design policy agrees with a welfare-maximizing design policy. The user's engage-versus-outside choice, the state dynamics, the retention rule, and subscription conversion all remain exactly as in the main text. The only change is the objective: welfare counts both the generative engine's payoff and the user's benefit from the same per-query choice problem.

\textbf{User benefit from the same per-query choice.}
For each displayed response $a\in\{\AIa,\AIf\}$, let
\[
u_{\user}(x,q,a)
\]
denote the user's per-query benefit from facing the same binary choice between engaging with the displayed response and taking the outside option. This is an extra primitive for the welfare extension only; it does not change the user's behavior, the state dynamics, or the revenue objective.

\textbf{Welfare flow and value function.}
For pre-subscription users, define the per-query welfare flow under action $a$ by
\begin{equation}
w_a(x,q):=r_a(x,q)+u_{\user}(x,q,a).
\end{equation}
After subscription, let
\[
w_{\sub}(x,p)
\]
denote the per-period welfare flow in the paid tier, combining the generative engine's subscription payoff and the user's value from the paid experience.

Fix $x$. Let $W_x(s,c,z)$ denote the maximal expected discounted welfare starting from state $(s,c,z)$. In the subscribed state,
\begin{equation}
W_x(s,c,1)
=
w_{\sub}(x,p)+\beta W_x(s,c,1)
=
\frac{w_{\sub}(x,p)}{1-\beta}.
\end{equation}
In the pre-subscription state,
\begin{align}
W_x(s,c,0)
&=
\mathbb E_{q\sim \mathcal P(\cdot\mid x)}
\Bigg[
\max_{a\in\{\AIa,\AIf\}}
\Big\{
w_a(x,q)
\nonumber
\\
&\quad\quad+
\beta\,
\mathbb E_{Y\sim \Pr(\cdot\mid x,q,a)}
\Big[
\rho_x\big(S^+_{a,Y},C^+_{a,Y},Z^+_{a,Y}\big)\,
W_x\big(S^+_{a,Y},C^+_{a,Y},Z^+_{a,Y}\big)
\Big]
\Big\}
\Bigg].
\label{eq:sw_bellman}
\end{align}
Relative to the revenue Bellman equation, the only change is that the current-period flow $r_a(x,q)$ is replaced by $w_a(x,q)$ and the subscribed-state flow is replaced by $w_{\sub}(x,p)$.

\begin{assumption}[Bounded welfare flows]\label{assm:sw_bounded}
There exist constants $U_{\max},\Delta_{\sub}^{\max}<\infty$ such that
\[
|u_{\user}(x,q,a)|\le U_{\max}
\qquad\text{for all }(x,q,a),
\]
and
\[
|w_{\sub}(x,p)-r_{\sub}(p)|\le \Delta_{\sub}^{\max}
\qquad\text{for all }x.
\]
\end{assumption}

Under \Cref{assm:regularity,assm:sw_bounded}, the same contraction argument as in \Cref{proof:opt_gate} yields a unique bounded welfare value function $W_x$.

\textbf{Welfare action values and the welfare edge.}
For $a\in\{\AIa,\AIf\}$, define
\begin{equation}
Q_x^{a,\sw}(s,c,q)
:=
w_a(x,q)
+
\beta\,
\mathbb E_{Y\sim \Pr(\cdot\mid x,q,a)}
\Big[
\rho_x\big(S^+_{a,Y},C^+_{a,Y},Z^+_{a,Y}\big)\,
W_x\big(S^+_{a,Y},C^+_{a,Y},Z^+_{a,Y}\big)
\Big].
\end{equation}
Define the welfare edge of \AIa\ over \AIf\ by
\begin{equation}
\Delta_x^{\sw}(s,c,q)
:=
Q_x^{\AIa,\sw}(s,c,q)-Q_x^{\AIf,\sw}(s,c,q).
\end{equation}
As in the main text, we break ties in favor of \AIa.

\begin{proposition}[The welfare design policy follows the same sign rule]
\label{prop:sw_sign_rule}
Under \Cref{assm:regularity,assm:sw_bounded}, the welfare-maximizing design policy chooses \AIa\ at $(s,c,q)$ if and only if
\[
\Delta_x^{\sw}(s,c,q)\ge 0.
\]
Equivalently, it chooses \AIf\ if and only if
\[
\Delta_x^{\sw}(s,c,q)<0.
\]
\end{proposition}

\begin{proof}
By definition, the welfare-maximizing design policy chooses an action that maximizes the welfare action value at $(s,c,q)$. Since the action set is binary and ties are broken toward \AIa, the policy chooses \AIa\ if and only if
\[
Q_x^{\AIa,\sw}(s,c,q)\ge Q_x^{\AIf,\sw}(s,c,q),
\]
which is equivalent to $\Delta_x^{\sw}(s,c,q)\ge 0$.
\end{proof}

This result shows that the welfare objective keeps the same point-wise comparison as the revenue objective. Welfare can change the preferred action only by changing the size of the \AIa\ edge, not by changing the basic decision rule.

\begin{proposition}[Welfare and revenue agree away from small revenue margins]
\label{prop:welfare_alignment_band}
Under \Cref{assm:regularity,assm:sw_bounded}, let
\[
\epsilon_{\sw}:=\max\big\{U_{\max},\Delta_{\sub}^{\max}\big\}.
\]
Then
\begin{equation}
\big\|W_x-V_x\big\|_{\infty}
\le
\frac{\epsilon_{\sw}}{1-\beta}.
\label{eq:sw_value_gap}
\end{equation}
Moreover,
\begin{equation}
\sup_{s,c,q}
\big|
\Delta_x^{\sw}(s,c,q)-\Delta_x(s,c,q)
\big|
\le
2U_{\max}+\frac{2\beta}{1-\beta}\,\epsilon_{\sw}.
\label{eq:sw_edge_gap}
\end{equation}
In particular, if at some $(s,c,q)$,
\begin{equation}
\big|\Delta_x(s,c,q)\big|
>
2U_{\max}+\frac{2\beta}{1-\beta}\,\epsilon_{\sw},
\label{eq:sw_margin_condition}
\end{equation}
then the welfare-maximizing and revenue-maximizing design policies choose the same action at $(s,c,q)$.
\end{proposition}

\begin{proof}
Let $\mathcal T_x$ denote the revenue Bellman operator from \Cref{proof:opt_gate}, and let $\mathcal T_x^{\sw}$ denote the welfare Bellman operator defined by \eqref{eq:sw_bellman}.

Fix any bounded function $F$. For pre-subscription states,
\begin{align*}
&\big|(\mathcal T_x^{\sw}F)(s,c,0)-(\mathcal T_xF)(s,c,0)\big|
\\
&\le
\mathbb E_{q\sim \mathcal P(\cdot\mid x)}
\Big[
\max_{a\in\{\AIa,\AIf\}}
\big|u_{\user}(x,q,a)\big|
\Big]
\\
&\le U_{\max},
\end{align*}
because the two operators differ only through the current-period user-benefit term. For subscribed states,
\[
\big|(\mathcal T_x^{\sw}F)(s,c,1)-(\mathcal T_xF)(s,c,1)\big|
=
\big|w_{\sub}(x,p)-r_{\sub}(p)\big|
\le
\Delta_{\sub}^{\max}.
\]
Hence
\[
\big\|\mathcal T_x^{\sw}F-\mathcal T_xF\big\|_{\infty}
\le
\epsilon_{\sw}.
\]

Using the fixed-point identities $W_x=\mathcal T_x^{\sw}W_x$ and $V_x=\mathcal T_xV_x$, together with the $\beta$-contraction property of $\mathcal T_x^{\sw}$, we obtain
\begin{align*}
\big\|W_x-V_x\big\|_{\infty}
&=
\big\|\mathcal T_x^{\sw}W_x-\mathcal T_xV_x\big\|_{\infty}
\\
&\le
\big\|\mathcal T_x^{\sw}W_x-\mathcal T_x^{\sw}V_x\big\|_{\infty}
+
\big\|\mathcal T_x^{\sw}V_x-\mathcal T_xV_x\big\|_{\infty}
\\
&\le
\beta\big\|W_x-V_x\big\|_{\infty}+\epsilon_{\sw},
\end{align*}
which yields \eqref{eq:sw_value_gap}.

Next fix $(s,c,q)$ and $a\in\{\AIa,\AIf\}$. Since the welfare and revenue problems share the same transition law,
\begin{align*}
&Q_x^{a,\sw}(s,c,q)-Q_x^{a}(s,c,q)
\\
&=
u_{\user}(x,q,a)
+
\beta\,
\mathbb E_{Y\sim \Pr(\cdot\mid x,q,a)}
\Big[
\rho_x\big(S^+_{a,Y},C^+_{a,Y},Z^+_{a,Y}\big)
\nonumber
\\
&\qquad\qquad\times
\big(
W_x\big(S^+_{a,Y},C^+_{a,Y},Z^+_{a,Y}\big)
-V_x\big(S^+_{a,Y},C^+_{a,Y},Z^+_{a,Y}\big)
\big)
\Big].
\end{align*}
Therefore,
\[
\big|Q_x^{a,\sw}(s,c,q)-Q_x^{a}(s,c,q)\big|
\le
U_{\max}+\beta\big\|W_x-V_x\big\|_{\infty}.
\]
Applying this bound to both actions gives
\begin{align*}
&\big|\Delta_x^{\sw}(s,c,q)-\Delta_x(s,c,q)\big|
\\
&\le
\big|Q_x^{\AIa,\sw}(s,c,q)-Q_x^{\AIa}(s,c,q)\big|
+
\big|Q_x^{\AIf,\sw}(s,c,q)-Q_x^{\AIf}(s,c,q)\big|
\\
&\le
2U_{\max}+2\beta\big\|W_x-V_x\big\|_{\infty}
\\
&\le
2U_{\max}+\frac{2\beta}{1-\beta}\,\epsilon_{\sw},
\end{align*}
which proves \eqref{eq:sw_edge_gap}.

Finally, if \eqref{eq:sw_margin_condition} holds, then \eqref{eq:sw_edge_gap} implies that $\Delta_x^{\sw}(s,c,q)$ has the same sign as $\Delta_x(s,c,q)$. Because both policies break ties toward \AIa, they choose the same action at $(s,c,q)$.
\end{proof}

This result localizes where welfare can matter. If the revenue edge already strongly favors one action, adding user-side value does not overturn the decision. Disagreement can arise only in states where the generative engine is already close to indifferent between \AIa\ and \AIf.

\subsection{Learning DP Primitives and Robustness of the Plug-in Design Policy}
\label{sec:dp-learning}

This subsection shows how the dynamic program can be implemented with logged data while staying fully aligned with the main framework. We keep the same state variables, the same engage-versus-outside choice, and the same design policy problem. The only goal here is to clarify which primitives must be learned and how estimation error affects the learned design policy.

\textbf{Logged data and overlap.}
Fix a user context $x$ and focus on pre-subscription decisions. Suppose we observe logged tuples
\[
(x_t,\; q_t,\; s_t,\; c_t,\; a_t,\; Y_t,\; \Pi_t,\; S_{t+1},\; C_{t+1},\; Z_{t+1},\; R_{t+1}),
\]
where $a_t\in\{\AIa,\AIf\}$ is the displayed response, $Y_t\in\{0,1\}$ is the user's engage-versus-outside decision, $\Pi_t$ is the realized one-period payoff to the generative engine, and $R_{t+1}\in\{0,1\}$ indicates whether the user returns in period $t+1$. The post-update state $(S_{t+1},C_{t+1},Z_{t+1})$ can be observed directly or reconstructed from $(q_t,s_t,c_t,a_t,Y_t)$ using the model's state-update rules.

\begin{assumption}[Exploration/overlap]
\label{assm:overlap}
There exists $\eta\in(0,1/2]$ such that for all $(x,q,s,c)$ in the support,
\begin{align}
&\mathbb{P}(a_t=\AIa\mid x_t=x,q_t=q,s_t=s,c_t=c)\ge \eta,
\nonumber
\\
\qquad
&\mathbb{P}(a_t=\AIf\mid x_t=x,q_t=q,s_t=s,c_t=c)\ge \eta.
\end{align}
\end{assumption}

\textbf{Primitives to estimate.}
The Bellman equation uses three one-step objects:
\begin{align}
m_a(x,q)
&:=
\mathbb{P}(Y_t=1\mid x_t=x,\; q_t=q,\; a_t=a),
\label{eq:learn_m_def}
\\
r_a(x,q)
&:=
\mathbb{E}\!\left[\Pi_t \mid x_t=x,\; q_t=q,\; a_t=a\right],
\label{eq:learn_r_def}
\\
\rho_x(s,c,z)
&:=
\mathbb{P}\!\left(R_{t+1}=1 \mid x_t=x,\; S_{t+1}=s,\; C_{t+1}=c,\; Z_{t+1}=z\right).
\label{eq:learn_rho_def}
\end{align}
Here $m_a(x,q)$ is the engage probability from the main text, $r_a(x,q)$ is the generative engine's expected one-period payoff, and $\rho_x(s,c,z)$ is the retention rule evaluated at the post-update state. We treat the query distribution $\mathcal P(\cdot\mid x)$, the state-update mappings in \eqref{eq:post_update_mappings_dp}, and the subscription rule as known model primitives. If those objects are also estimated, the same arguments below continue to hold after adding their error terms.

Let $\widehat m_a$, $\widehat r_a$, and $\widehat \rho_x$ denote fitted estimates of \eqref{eq:learn_m_def}--\eqref{eq:learn_rho_def}. Define the fitted engage probabilities
\[
\widehat \pi_a(1\mid x,q):=\widehat m_a(x,q),
\qquad
\widehat \pi_a(0\mid x,q):=1-\widehat m_a(x,q),
\]
and analogously $\pi_a(1\mid x,q):=m_a(x,q)$ and $\pi_a(0\mid x,q):=1-m_a(x,q)$.

\textbf{Plug-in Bellman equation and learned design policy.}
For any bounded function $F$, define the estimated Bellman operator by
\begin{align}
(\widehat{\mathcal T}_x F)(s,c,1)
&:=
r_{\sub}(p)+\beta F(s,c,1),
\nonumber
\\
(\widehat{\mathcal T}_x F)(s,c,0)
&:=
\mathbb E_{q\sim \mathcal P(\cdot\mid x)}
\Bigg[
\max_{a\in\{\AIa,\AIf\}}
\Big\{
\widehat r_a(x,q)
\nonumber
\\
&\qquad\qquad
+\beta
\sum_{y\in\{0,1\}}
\widehat \pi_a(y\mid x,q)\,
\widehat \rho_x\big(S^+_{a,y},C^+_{a,y},Z^+_{a,y}\big)\,
F\big(S^+_{a,y},C^+_{a,y},Z^+_{a,y}\big)
\Big\}
\Bigg].
\label{eq:learn_bellman}
\end{align}
Under admissible fitted primitives, $\widehat{\mathcal T}_x$ is a $\beta$-contraction and therefore has a unique bounded fixed point, denoted by $\widehat V_x$.

Define the plug-in action values by
\begin{align}
\widehat Q_x^{a}(s,c,q)
&:=
\widehat r_a(x,q)
\nonumber
\\
&\quad
+\beta
\sum_{y\in\{0,1\}}
\widehat \pi_a(y\mid x,q)\,
\widehat \rho_x\big(S^+_{a,y},C^+_{a,y},Z^+_{a,y}\big)\,
\widehat V_x\big(S^+_{a,y},C^+_{a,y},Z^+_{a,y}\big),
\label{eq:learn_Qhat_def}
\end{align}
and the plug-in \AIa\ edge by
\begin{equation}
\widehat \Delta_x(s,c,q)
:=
\widehat Q_x^{\AIa}(s,c,q)-\widehat Q_x^{\AIf}(s,c,q).
\label{eq:learn_deltahat_def}
\end{equation}
The learned design policy uses the same tie-breaking rule as the main text:
\begin{equation}
\widehat g_x(s,c,q)=\AIa
\quad\Longleftrightarrow\quad
\widehat \Delta_x(s,c,q)\ge 0.
\label{eq:learn_policy_def}
\end{equation}

\begin{proposition}[Plug-in estimation error yields small value and policy error]
\label{prop:learn_robustness}
Fix $x$. Suppose the fitted primitives satisfy
\[
0\le \widehat m_a(x,q)\le 1,
\qquad
0\le \widehat \rho_x(s,c,z)\le 1,
\]
and that there exists $R_{\max}<\infty$ such that
\[
\sup_{a,q}\big|r_a(x,q)\big|
\le R_{\max},
\qquad
\sup_{a,q}\big|\widehat r_a(x,q)\big|
\le R_{\max},
\qquad
\big|r_{\sub}(p)\big|\le R_{\max}.
\]
Define
\begin{align}
\epsilon_m
&:=
\sup_{a,q}
\big|\widehat m_a(x,q)-m_a(x,q)\big|,
\nonumber
\\
\epsilon_r
&:=
\sup_{a,q}
\big|\widehat r_a(x,q)-r_a(x,q)\big|,
\nonumber
\\
\epsilon_{\rho}
&:=
\sup_{s,c,z}
\big|\widehat \rho_x(s,c,z)-\rho_x(s,c,z)\big|,
\label{eq:learn_eps_primitives}
\end{align}
and let
\[
V_{\max}:=\frac{R_{\max}}{1-\beta},
\qquad
B_x:=\epsilon_r+\beta\big(2\epsilon_m+\epsilon_{\rho}\big)V_{\max}.
\]
Then the following hold:
\begin{enumerate}[noitemsep, topsep=0pt]
\item
\[
\big\|\widehat V_x-V_x\big\|_{\infty}
\le
\frac{B_x}{1-\beta}.
\label{eq:learn_value_bound}
\]
\item
\[
\sup_{s,c,q}
\big|
\widehat \Delta_x(s,c,q)-\Delta_x(s,c,q)
\big|
\le
\frac{2B_x}{1-\beta}.
\label{eq:learn_edge_bound}
\]
\item If at some $(s,c,q)$,
\[
\big|\Delta_x(s,c,q)\big|>\frac{2B_x}{1-\beta},
\label{eq:learn_margin_condition}
\]
then the learned and optimal design policies agree at $(s,c,q)$:
\[
\widehat g_x(s,c,q)=g_x^*(s,c,q).
\label{eq:learn_policy_stabilize}
\]
\item Let $V_x^{\widehat g}$ denote the value induced by following the learned design policy $\widehat g_x$ thereafter. Then for all pre-subscription states $(s,c,0)$,
\[
0
\le
V_x(s,c,0)-V_x^{\widehat g}(s,c,0)
\le
\frac{2B_x}{(1-\beta)^2}.
\label{eq:learn_regret}
\]
\end{enumerate}
\end{proposition}

\begin{proof}
Let $\mathcal T_x$ denote the true Bellman operator. Because both the true and fitted retention probabilities lie in $[0,1]$, both $\mathcal T_x$ and $\widehat{\mathcal T}_x$ are $\beta$-contractions on bounded functions.

First, by the bounded payoff assumption,
\[
\|V_x\|_{\infty}\le V_{\max}.
\]
Now fix any bounded $F$ with $\|F\|_{\infty}\le V_{\max}$. For a pre-subscription state $(s,c,0)$, the difference between the fitted and true one-step maximands is bounded by the largest action-level discrepancy. For any $a\in\{\AIa,\AIf\}$ and query $q$,
\begin{align*}
&\Big|
\widehat r_a(x,q)
+
\beta
\sum_{y}
\widehat \pi_a(y\mid x,q)\,
\widehat \rho_x\big(S^+_{a,y},C^+_{a,y},Z^+_{a,y}\big)\,
F\big(S^+_{a,y},C^+_{a,y},Z^+_{a,y}\big)
\\
&\qquad
-
r_a(x,q)
-
\beta
\sum_{y}
\pi_a(y\mid x,q)\,
\rho_x\big(S^+_{a,y},C^+_{a,y},Z^+_{a,y}\big)\,
F\big(S^+_{a,y},C^+_{a,y},Z^+_{a,y}\big)
\Big|
\\
&\le
\epsilon_r
+
\beta
\sum_{y}
\widehat \pi_a(y\mid x,q)\,
\Big|
\widehat \rho_x\big(S^+_{a,y},C^+_{a,y},Z^+_{a,y}\big)
-
\rho_x\big(S^+_{a,y},C^+_{a,y},Z^+_{a,y}\big)
\Big|
\cdot
\|F\|_{\infty}
\\
&\qquad
+
\beta
\Big|
\sum_{y}
\big(\widehat \pi_a(y\mid x,q)-\pi_a(y\mid x,q)\big)\,
\rho_x\big(S^+_{a,y},C^+_{a,y},Z^+_{a,y}\big)\,
F\big(S^+_{a,y},C^+_{a,y},Z^+_{a,y}\big)
\Big|
\\
&\le
\epsilon_r+\beta\epsilon_{\rho}V_{\max}
+
\beta
\big|
\widehat m_a(x,q)-m_a(x,q)
\big|
\cdot
\big|
g_{a,1}(q)-g_{a,0}(q)
\big|,
\end{align*}
where
\[
g_{a,y}(q)
:=
\rho_x\big(S^+_{a,y},C^+_{a,y},Z^+_{a,y}\big)\,
F\big(S^+_{a,y},C^+_{a,y},Z^+_{a,y}\big).
\]
Since $|g_{a,y}(q)|\le \|F\|_{\infty}\le V_{\max}$, we have
\[
\big|g_{a,1}(q)-g_{a,0}(q)\big|\le 2V_{\max},
\]
so the last display is bounded by
\[
\epsilon_r+\beta\big(2\epsilon_m+\epsilon_{\rho}\big)V_{\max}
=
B_x.
\]
Taking the maximum over $a$ and the expectation over $q$ yields
\[
\big\|\widehat{\mathcal T}_xF-\mathcal T_xF\big\|_{\infty}
\le
B_x.
\]

Using the fixed-point identities $\widehat V_x=\widehat{\mathcal T}_x\widehat V_x$ and $V_x=\mathcal T_xV_x$, together with the contraction property of $\widehat{\mathcal T}_x$, we obtain
\begin{align*}
\big\|\widehat V_x-V_x\big\|_{\infty}
&=
\big\|\widehat{\mathcal T}_x\widehat V_x-\mathcal T_xV_x\big\|_{\infty}
\\
&\le
\big\|\widehat{\mathcal T}_x\widehat V_x-\widehat{\mathcal T}_xV_x\big\|_{\infty}
+
\big\|\widehat{\mathcal T}_xV_x-\mathcal T_xV_x\big\|_{\infty}
\\
&\le
\beta\big\|\widehat V_x-V_x\big\|_{\infty}+B_x,
\end{align*}
which proves \eqref{eq:learn_value_bound}.

Next fix $(s,c,q)$ and $a\in\{\AIa,\AIf\}$. By adding and subtracting terms exactly as above,
\begin{align*}
\big|\widehat Q_x^{a}(s,c,q)-Q_x^{a}(s,c,q)\big|
&\le
\epsilon_r
+
\beta\big\|\widehat V_x-V_x\big\|_{\infty}
+
\beta\big(2\epsilon_m+\epsilon_{\rho}\big)V_{\max}
\\
&=
B_x+\beta\big\|\widehat V_x-V_x\big\|_{\infty}
\\
&\le
\frac{B_x}{1-\beta}.
\end{align*}
Applying this bound to both actions gives
\begin{align*}
\big|
\widehat \Delta_x(s,c,q)-\Delta_x(s,c,q)
\big|
&\le
\big|\widehat Q_x^{\AIa}(s,c,q)-Q_x^{\AIa}(s,c,q)\big|
+
\big|\widehat Q_x^{\AIf}(s,c,q)-Q_x^{\AIf}(s,c,q)\big|
\\
&\le
\frac{2B_x}{1-\beta},
\end{align*}
which proves \eqref{eq:learn_edge_bound}.

If \eqref{eq:learn_margin_condition} holds, then \eqref{eq:learn_edge_bound} implies that $\widehat \Delta_x(s,c,q)$ has the same sign as $\Delta_x(s,c,q)$. Because the learned and optimal design policies use the same tie-breaking rule in favor of \AIa, they choose the same action, proving \eqref{eq:learn_policy_stabilize}.

Finally, let
\[
\epsilon_{\Delta}:=
\sup_{s,c,q}
\big|
\widehat \Delta_x(s,c,q)-\Delta_x(s,c,q)
\big|.
\]
At any query-state point where the learned and optimal policies disagree, the same-sign argument above implies
\[
\big|\Delta_x(s,c,q)\big|\le \epsilon_{\Delta}.
\]
Because the action set is binary, the loss from choosing $\widehat g_x$ instead of $g_x^*$ at that point is at most $\epsilon_{\Delta}$ in true action value. Therefore, for every pre-subscription state,
\[
V_x(s,c,0)
\le
(\mathcal T_x^{\widehat g}V_x)(s,c,0)+\epsilon_{\Delta},
\]
where $\mathcal T_x^{\widehat g}$ is the Bellman operator induced by the learned design policy. Since $V_x^{\widehat g}=\mathcal T_x^{\widehat g}V_x^{\widehat g}$ and $\mathcal T_x^{\widehat g}$ is also a $\beta$-contraction,
\begin{align*}
\big\|V_x-V_x^{\widehat g}\big\|_{\infty}
&\le
\epsilon_{\Delta}
+
\beta\big\|V_x-V_x^{\widehat g}\big\|_{\infty},
\end{align*}
which yields
\[
\big\|V_x-V_x^{\widehat g}\big\|_{\infty}
\le
\frac{\epsilon_{\Delta}}{1-\beta}
\le
\frac{2B_x}{(1-\beta)^2}.
\]
This proves \eqref{eq:learn_regret}.
\end{proof}

This proposition shows that learning matters most near indifference. If the true \AIa\ edge is already large in magnitude, moderate estimation error does not change the design choice. Estimation error can overturn the policy only in states where the generative engine is already close to indifferent between \AIa\ and \AIf.

\begin{corollary}[Consistency of the learned design policy]
\label{cor:learn_consistency}
Fix $x$ and consider a sequence of fitted primitives
\[
\big(\widehat m_{a,n},\widehat r_{a,n},\widehat \rho_{x,n}\big)_{n\ge 1}
\]
satisfying the admissibility conditions in \Cref{prop:learn_robustness}. If
\[
\epsilon_{m,n}\to 0,
\qquad
\epsilon_{r,n}\to 0,
\qquad
\epsilon_{\rho,n}\to 0,
\]
then the associated plug-in value functions, action edges, and learned design policies satisfy
\begin{align}
\big\|\widehat V_{x,n}-V_x\big\|_{\infty}
&\to 0,
\label{eq:learn_value_consistency}
\\
\sup_{s,c,q}
\big|
\widehat \Delta_{x,n}(s,c,q)-\Delta_x(s,c,q)
\big|
&\to 0,
\label{eq:learn_edge_consistency}
\\
\sup_{s,c}
\big(
V_x(s,c,0)-V_x^{\widehat g_n}(s,c,0)
\big)
&\to 0.
\label{eq:learn_policy_value_consistency}
\end{align}
Moreover, for any fixed $(s,c,q)$ such that
\[
\Delta_x(s,c,q)\neq 0,
\]
the learned and optimal design policies eventually agree at that point:
\[
\widehat g_{x,n}(s,c,q)=g_x^*(s,c,q)
\qquad
\text{for all sufficiently large }n.
\label{eq:learn_pointwise_consistency}
\]
\end{corollary}

\begin{proof}
Let
\[
B_{x,n}
:=
\epsilon_{r,n}
+
\beta\big(2\epsilon_{m,n}+\epsilon_{\rho,n}\big)V_{\max}.
\]
By assumption, $B_{x,n}\to 0$. Therefore \Cref{prop:learn_robustness} implies
\[
\big\|\widehat V_{x,n}-V_x\big\|_{\infty}
\le
\frac{B_{x,n}}{1-\beta}\to 0,
\]
which proves \eqref{eq:learn_value_consistency}, and
\[
\sup_{s,c,q}
\big|
\widehat \Delta_{x,n}(s,c,q)-\Delta_x(s,c,q)
\big|
\le
\frac{2B_{x,n}}{1-\beta}\to 0,
\]
which proves \eqref{eq:learn_edge_consistency}. The value bound \eqref{eq:learn_regret} from \Cref{prop:learn_robustness} then yields \eqref{eq:learn_policy_value_consistency}.

Finally, fix $(s,c,q)$ with $\Delta_x(s,c,q)\neq 0$. Since \eqref{eq:learn_edge_consistency} gives
\[
\widehat \Delta_{x,n}(s,c,q)\to \Delta_x(s,c,q),
\]
the sign of $\widehat \Delta_{x,n}(s,c,q)$ must eventually match the sign of $\Delta_x(s,c,q)$. Because both policies break ties toward \AIa, they eventually choose the same action, proving \eqref{eq:learn_pointwise_consistency}.
\end{proof}

This corollary gives the clean asymptotic implication. As the estimated primitives become uniformly accurate, the plug-in value function converges to the true value function, the learned \AIa\ edge converges to the true edge, and the learned design policy becomes value-optimal. Exact policy agreement is immediate away from knife-edge states where $\Delta_x(s,c,q)=0$.

\subsection{Extra Simulation Plots}
\label[appsec]{app:extra_simu}

\subsubsection{Optimal Policy Characterization}
\label[appsec]{app:extra_simu_policy}

\begin{figure}[h]
  \centering
  \begin{subfigure}[t]{0.49\linewidth}
    \centering
    \includegraphics[width=\linewidth]{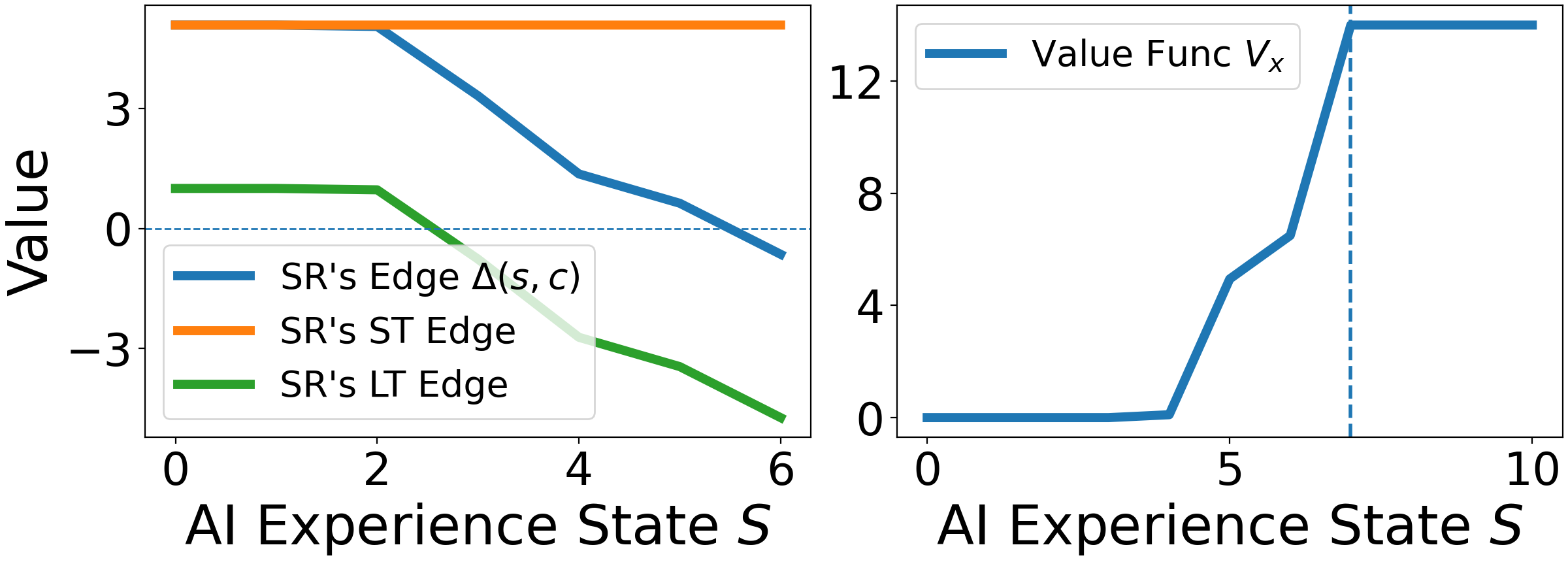}
    \caption{$\gamma=0.2$}
    \label{fig:type_pdf_query_g02}
  \end{subfigure}
  \hfill
  \begin{subfigure}[t]{0.49\linewidth}
    \centering
    \includegraphics[width=\linewidth]{archive-versions/PR-version/Plotting/fig_section41_threshold_gamma_08.png}
    \caption{$\gamma=0.8$}
    \label{fig:type_pdf_query_g08}
  \end{subfigure}

  \caption{
  Simulation illustration of the threshold rule in \Cref{prop:opt_gate} for a fixed context $x$ and \AIa \textit{experience state} $c$, under different user heterogeneity. 
  \emph{Left:} \textit{Ad sensitivity} $\gamma = 0.2$;
  \emph{Right:} \textit{AI sensitivity} $\gamma = 0.8$.
  Inside each subplot: \emph{Left:} Query-averaged $\AIa$ edge $\mathbb{E}_q\Delta_x(s,c,q)$ and its short-term (ST) and continuation-value (LT) components from \Cref{rem:tradeoff} as \textit{AI experience $s$} varies.
  \emph{Right:} pre-subscription value $V_x(s,c,0)$ as a function of $s$; the dashed line marks the subscription threshold $\tau(x,p,c)$, after which the value equals the subscribed level.
  }
  \label{fig:sim_threshold_gamma}
\end{figure}

\begin{figure}[h]
  \centering
  \begin{subfigure}[t]{0.49\linewidth}
    \centering
    \includegraphics[width=\linewidth]{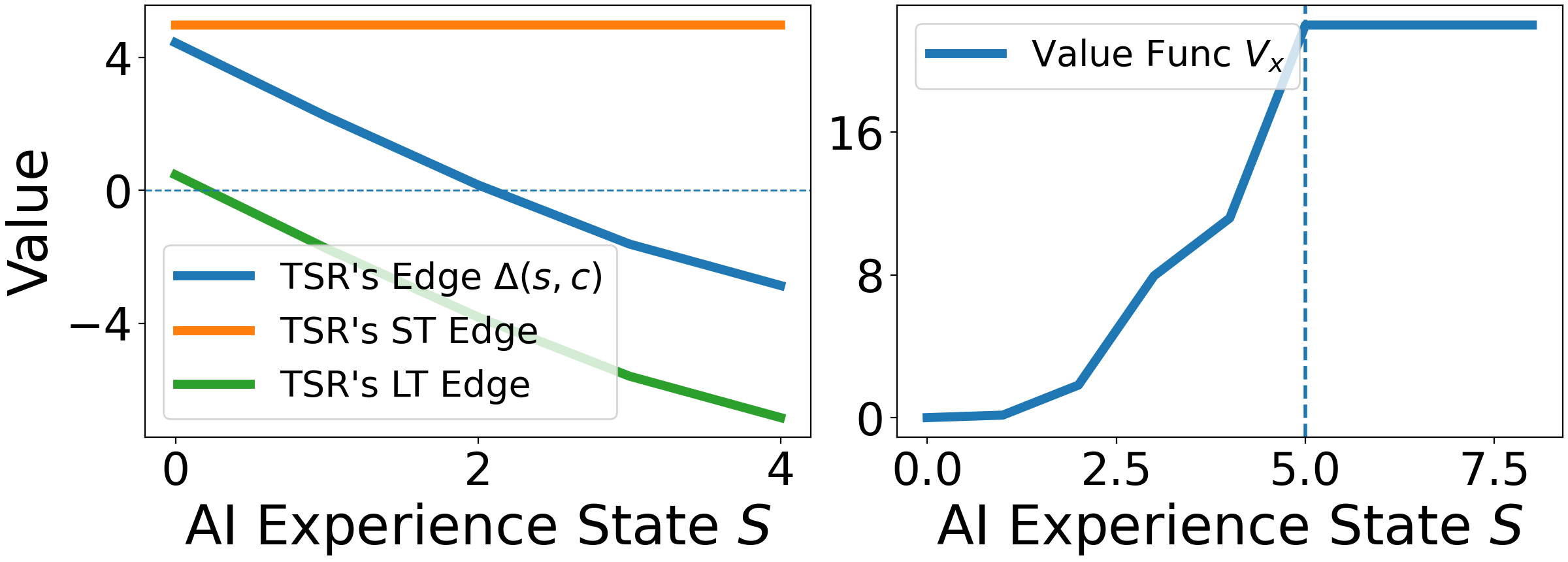}
    \caption{$\theta=0.2$}
    \label{fig:type_pdf_query_t02}
  \end{subfigure}
  \hfill
  \begin{subfigure}[t]{0.49\linewidth}
    \centering
    \includegraphics[width=\linewidth]{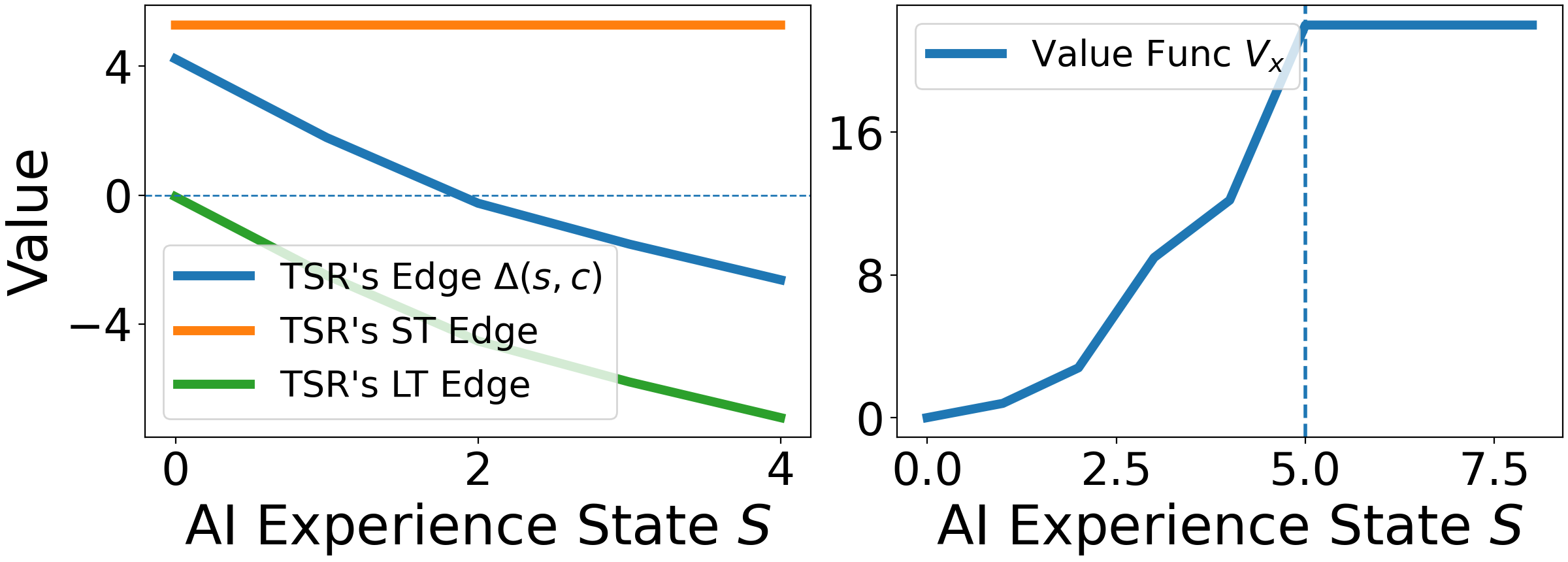}
    \caption{$\theta=0.8$}
    \label{fig:type_pdf_query_t08}
  \end{subfigure}

  \caption{
  Simulation illustration of the threshold rule in \Cref{prop:opt_gate} for a fixed context $x$ and \AIa \textit{experience state} $c$, under different user heterogeneity. 
  \emph{Left:} \textit{Ad reliance} $\theta = 0.2$;
  \emph{Right:} \textit{AI reliance} $\theta = 0.8$.
  Inside each subplot: \emph{Left:} Query-averaged $\AIa$ edge $\mathbb{E}_q\Delta_x(s,c,q)$ and its short-term (ST) and continuation-value (LT) components from \Cref{rem:tradeoff} as \textit{AI experience $s$} varies.
  \emph{Right:} pre-subscription value $V_x(s,c,0)$ as a function of $s$; the dashed line marks the subscription threshold $\tau(x,p,c)$, after which the value equals the subscribed level.
  }
  \label{fig:sim_threshold_theta}
\end{figure}

\subsubsection{Threshold Structure Under User-Query Types}
\label[appsec]{app:extra_simu_user_query}

\begin{figure}[h]
  \centering
  \includegraphics[width=0.50\linewidth]{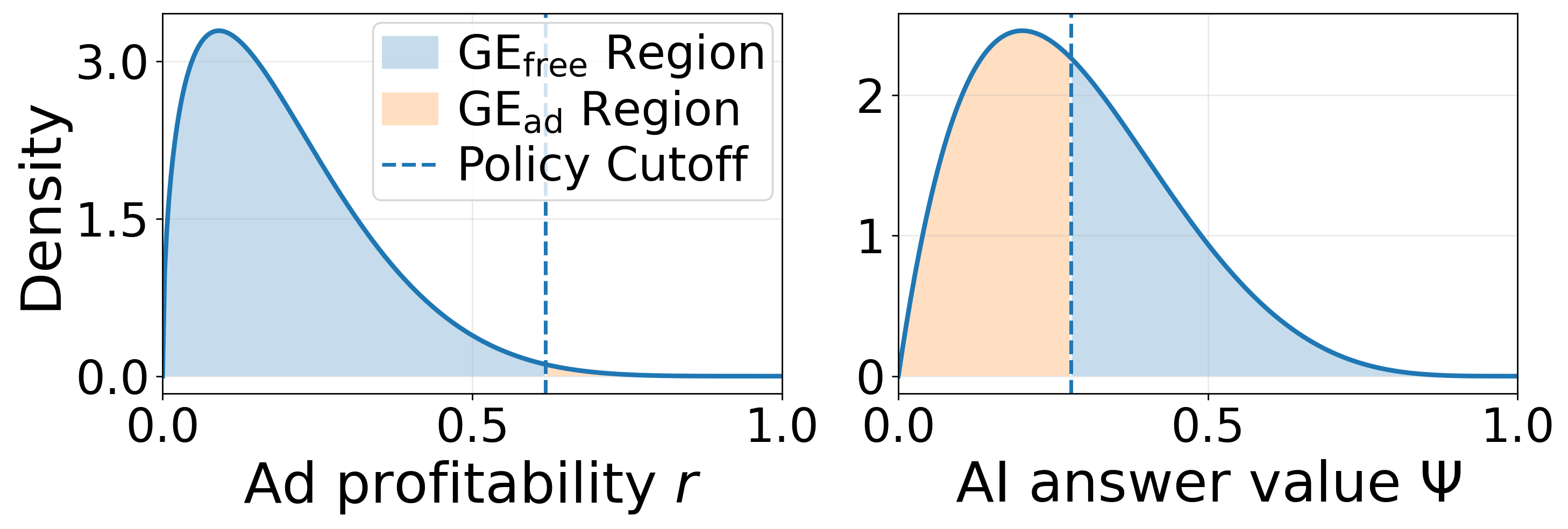}
  \vspace{-10pt}
  \caption{
  Query-type illustration with fixing user context $x$ and drawing query types $r, \psi$ from beta distributions.
  Shaded regions indicate which search outcome is optimal: $\AIf$ vs.\ $\AIa$.
  }
  \label{fig:type_pdf_query}
  \vspace{-10pt}
\end{figure}

\section{Notes for Experiment}\label{app:experiment_notes}
\subsection{Additional Experimental Setups}
\label[appsec]{app:sim_details}

\paragraph{Population, horizon, and shared randomness.}
We simulate one platform interacting with $N=500$ heterogeneous users over a horizon $T=20$ with discount factor $\beta=0.95$ (unless otherwise noted).
Each user $i$ draws a fixed latent type $\gamma_i$ i.i.d.\ from $\mathrm{Beta}$.
In each period $t$, the environment draws a query profitability signal $r_{it}$ and an AI-quality signal $\psi_{it}$ i.i.d.\ from $\mathrm{Unif}[0,1]^2$.
To make policy comparisons apples-to-apples, we evaluate all policies on the \emph{same} sampled population and the same realizations of
$\{(r_{it},\psi_{it})\}$, as well as the same uniform random draws used to realize consumption and retention outcomes.

\paragraph{User utilities and consumption.}
Given the platform’s displayed outcome $a_{it}\in\{\AIa,\AIf\}$, the user chooses whether to consume the displayed content versus an outside option (outside option utility normalized to $0$).
Consumption follows a logit rule with utility indices
\[
v_{\AIf}(\gamma;r,\psi)=w_\psi\,\psi+w_{r,\AIf}\,r,
\qquad
v_{\AIa}(\gamma;r,c)=u^{\AIa}_0+u^{\AIa}_r\,r-w_\gamma\,\gamma-u^{\AIa}_c\,c,
\]
where $c$ is the \AIa fatigue state defined below. (All coefficients are fixed environment parameters; we use the values in our code/config unless explicitly overridden.)

\paragraph{State dynamics and platform payoffs.}
Each pre-subscription user maintains two integer states $(s_{it},c_{it})$ initialized at $(0,0)$ and truncated to caps $(S_{\max},C_{\max})$.
If the user consumes \AIf, the AI experience state increases by a discrete increment
$s_{i,t+1}\leftarrow \min\{s_{it}+\Delta_s(\psi_{it}),S_{\max}\}$ with $\Delta_s(\psi)=1+\mathbf{1}\{\psi\ge \psi_{\mathrm{cut}}\}$.
If the user consumes \AIa, the fatigue state increases by one: $c_{i,t+1}\leftarrow \min\{c_{it}+1,C_{\max}\}$.
If the user does not consume, the state does not change.
Platform per-period payoff is earned only upon consumption: consuming \AIf yields a serving cost $-\kappa_{\mathrm{free}}$, while consuming \AIa yields ad revenue
$r_{\AIa}(r,c,s)=\max\{b_0+b_r\,r-b_c\,c-b_s\,s,0\}$ (allowing SR monetization to deteriorate with fatigue and, optionally, with accumulated AI experience).

\paragraph{Retention, conversion, and subscription value.}
Pre-subscription retention is state-dependent with a floor:
\[
\rho(s,c)=\rho_{\min}+(1-\rho_{\min})\cdot \sigma\!\big(\alpha_s\,s-\alpha_c\,c\big),
\]
where $\sigma(\cdot)$ is the logistic link.
Conversion is triggered deterministically when accumulated experience reaches a type-dependent cutoff:
a user converts once $s_{it}\ge \lceil \tau\rceil$, with $\tau(\theta, p, c)=\tau_0+\tau_p\cdot p-\tau_\theta \cdot \theta - \tau_c \cdot c$.
Upon conversion, the user enters an absorbing subscribed regime; we add the discounted continuation value starting next period,
$\beta^{t+1}V_{\sub}$ with $V_{\sub}=(p-\kappa_{\mathrm{paid}})/(1-\beta)$, and stop simulating further pre-subscription dynamics for that user.
In our “total active users” metric, we count both retained pre-subscription users and already-converted (subscribed) users as active.

\paragraph{DP computation and type binning.}
The Optimal DP policy is computed by value iteration on the integer grid $(s,c)$ in the pre-subscription regime.
Bellman expectations over query signals are approximated by Monte Carlo with $n_q$ i.i.d.\ draws per state using a fixed random seed.
To reduce the number of DP solves, we bin $\gamma$ onto a uniform grid over $[0,1]^2$ with $n_{\text{bins}}$ points per dimension and solve the DP only for the unique binned type pairs appearing in the simulated population; each user is mapped to the nearest grid point for policy lookup.

\subsection{Market-Condition Sensitivity Analysis}
\label[appsec]{app:ablation_details}

\paragraph{Market conditions.}
We conduct a paired high/low sensitivity analysis across four market dimensions.
For each dimension, we simulate both a ``high'' and ``low'' variant while holding all other parameters at baseline values:

\begin{enumerate}[noitemsep, topsep=2pt]
\item \textbf{Ad sensitivity} ($\gamma$): High $\gamma\sim\mathrm{Beta}(5,2)$ concentrates mass on ad-sensitive users; Low $\gamma\sim\mathrm{Beta}(2,5)$ concentrates mass on ad-tolerant users. Baseline: $\gamma\sim\mathrm{Unif}[0,1]$.

\item \textbf{Query profitability} ($r$): High $r\sim\mathrm{Beta}(5,2)$ shifts queries toward high ad revenue; Low $r\sim\mathrm{Beta}(2,5)$ reduces ad profitability. Baseline: $r\sim\mathrm{Unif}[0,1]$.

\item \textbf{Inference cost} ($\kappa_{\mathrm{free}}$): High $\kappa=1.3$ makes \AIf\ expensive to serve; Low $\kappa=0.5$ makes \AIf\ cheap. Baseline: $\kappa=0.9$.

\item \textbf{Outside competition} ($\omega$): High $\omega=0.5$ raises the outside option; Low $\omega=0.2$ weakens competition. Baseline: $\omega=0$.
\end{enumerate}

This yields $2 \times 4 = 8$ market conditions, each simulated under all four policies ($N=300$ users, $T=20$ periods, $\beta=0.95$, $n_{\mathrm{bins}}=5$, $n_q=60$).
To ensure statistical reliability, each configuration is run across five independent random seeds (seeds $42, 1042, 2042, 3042, 4042$), for a total of $8 \times 4 \times 5 = 160$ simulation runs.
All results report mean~$\pm$~standard deviation across seeds.

\paragraph{Detailed results.}
\Cref{tab:ablation_full} reports the final cumulative payoff (mean$\pm$std), \AIf\ exposure rate, and subscriber count under each condition.

\begin{table*}[h]
\centering
\caption{Full sensitivity analysis results (mean$\pm$std over 5 seeds, $T{=}20$).
\textit{DP \AIf\%} reports the average share of active non-subscribers shown \AIf\ by the Optimal DP policy.
\textit{Subs} reports the total subscribers under Optimal DP and Always \AIf\ policies.}
\label{tab:ablation_full}
\small
\begin{tabular}{l cccc cc cc}
\toprule
Condition & Opt.\ DP & Greedy & Always \AIa & Always \AIf & \AIf\% & Subs(DP) & Subs(\AIf) \\
\midrule
High $\gamma$ & \textbf{9.19}$\pm$0.35 & 6.66$\pm$0.44 & 1.40$\pm$0.02 & 8.42$\pm$0.50 & 92\% & 260$\pm$4 & 276$\pm$4 \\
Low $\gamma$ & \textbf{3.65}$\pm$0.18 & 1.85$\pm$0.05 & 1.41$\pm$0.04 & 2.27$\pm$0.17 & 87\% & 217$\pm$6 & 252$\pm$6 \\
\addlinespace
High $r$ & \textbf{3.85}$\pm$0.22 & 3.11$\pm$0.17 & 1.83$\pm$0.02 & 2.34$\pm$0.27 & 87\% & 130$\pm$8 & 156$\pm$8 \\
Low $r$ & \textbf{7.22}$\pm$0.15 & 5.03$\pm$0.14 & 0.82$\pm$0.02 & 7.04$\pm$0.14 & 92\% & 285$\pm$2 & 292$\pm$2 \\
\addlinespace
High $\kappa$ & \textbf{5.13}$\pm$0.20 & 3.07$\pm$0.12 & 1.36$\pm$0.02 & 3.62$\pm$0.28 & 64\% & 196$\pm$10 & 265$\pm$6 \\
Low $\kappa$ & \textbf{8.16}$\pm$0.28 & 6.03$\pm$0.27 & 1.36$\pm$0.02 & 7.30$\pm$0.29 & 93\% & 243$\pm$6 & 265$\pm$6 \\
\addlinespace
High $\omega$ & \textbf{5.20}$\pm$0.28 & 3.77$\pm$0.25 & 1.34$\pm$0.01 & 4.16$\pm$0.30 & 84\% & 196$\pm$7 & 226$\pm$6 \\
Low $\omega$ & \textbf{6.03}$\pm$0.22 & 4.10$\pm$0.26 & 1.35$\pm$0.02 & 4.97$\pm$0.26 & 86\% & 223$\pm$6 & 251$\pm$5 \\
\bottomrule
\end{tabular}
\end{table*}

\paragraph{Detailed trajectory plots.}
\Cref{fig:detail_gamma,fig:detail_r,fig:detail_kappa,fig:detail_omega} show the full trajectory plots for all four metrics---cumulative payoff, active users, subscribers, and AI exposure rate---under each market condition.

\begin{figure}[h]
\centering
\includegraphics[width=\linewidth]{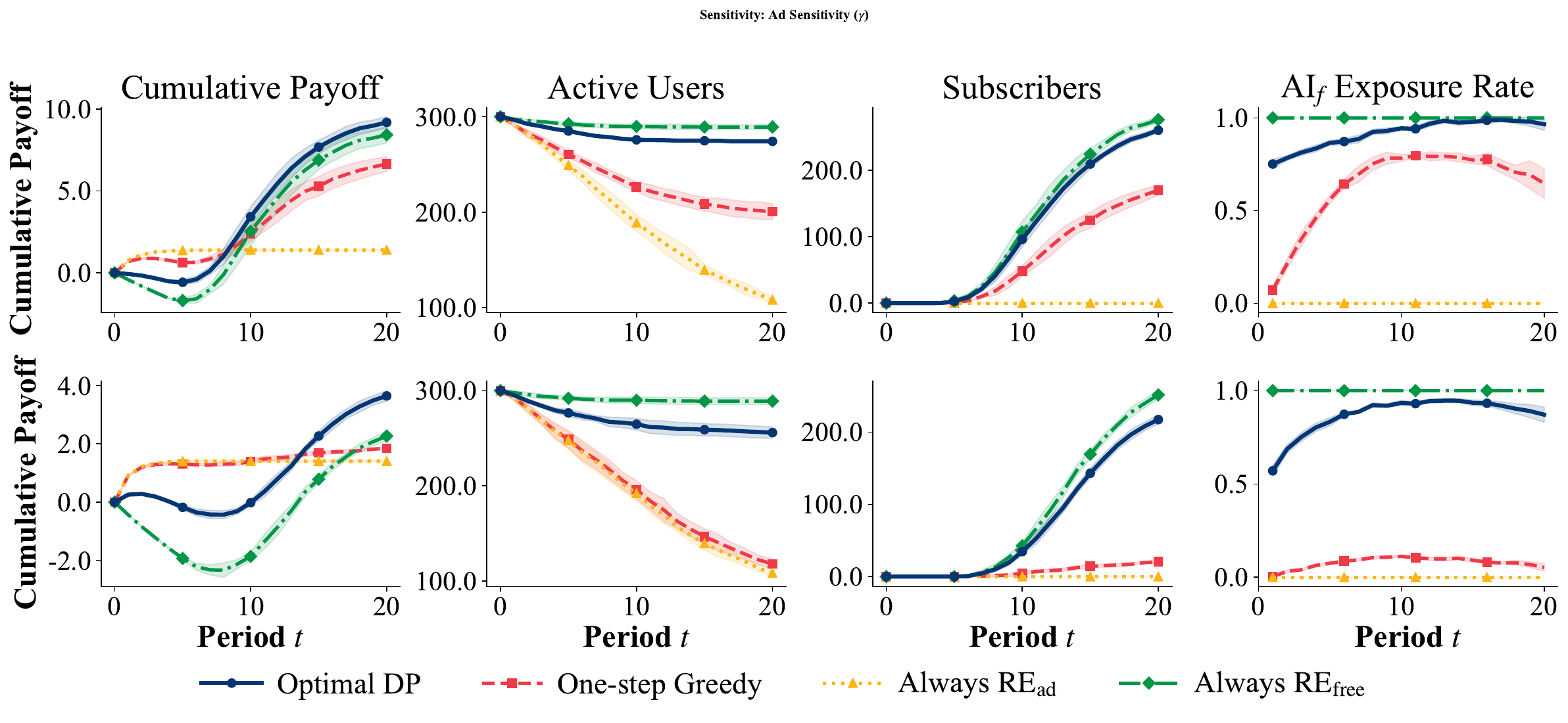}
\caption{Sensitivity to ad sensitivity ($\gamma$): all four metrics under high-$\gamma$ (top) and low-$\gamma$ (bottom) conditions.
High ad sensitivity accelerates subscription conversion and increases the DP's payoff advantage.}
\label{fig:detail_gamma}
\end{figure}

\begin{figure}[h]
\centering
\includegraphics[width=\linewidth]{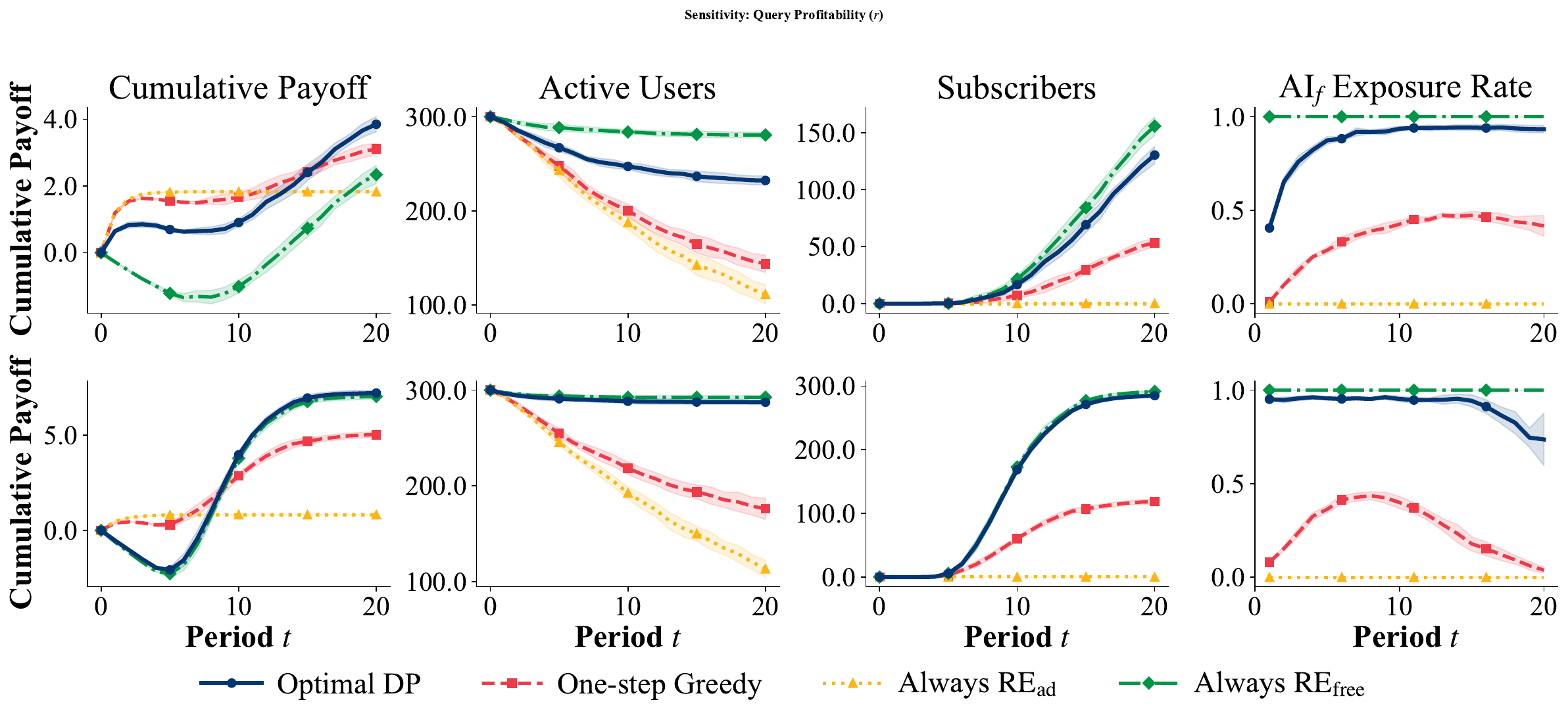}
\caption{Sensitivity to query profitability ($r$): all four metrics under high-$r$ (top) and low-$r$ (bottom) conditions.
Low-profit queries reduce the opportunity cost of \AIf, enabling the DP to invest more aggressively in subscription conversion.}
\label{fig:detail_r}
\end{figure}

\begin{figure}[h]
\centering
\includegraphics[width=\linewidth]{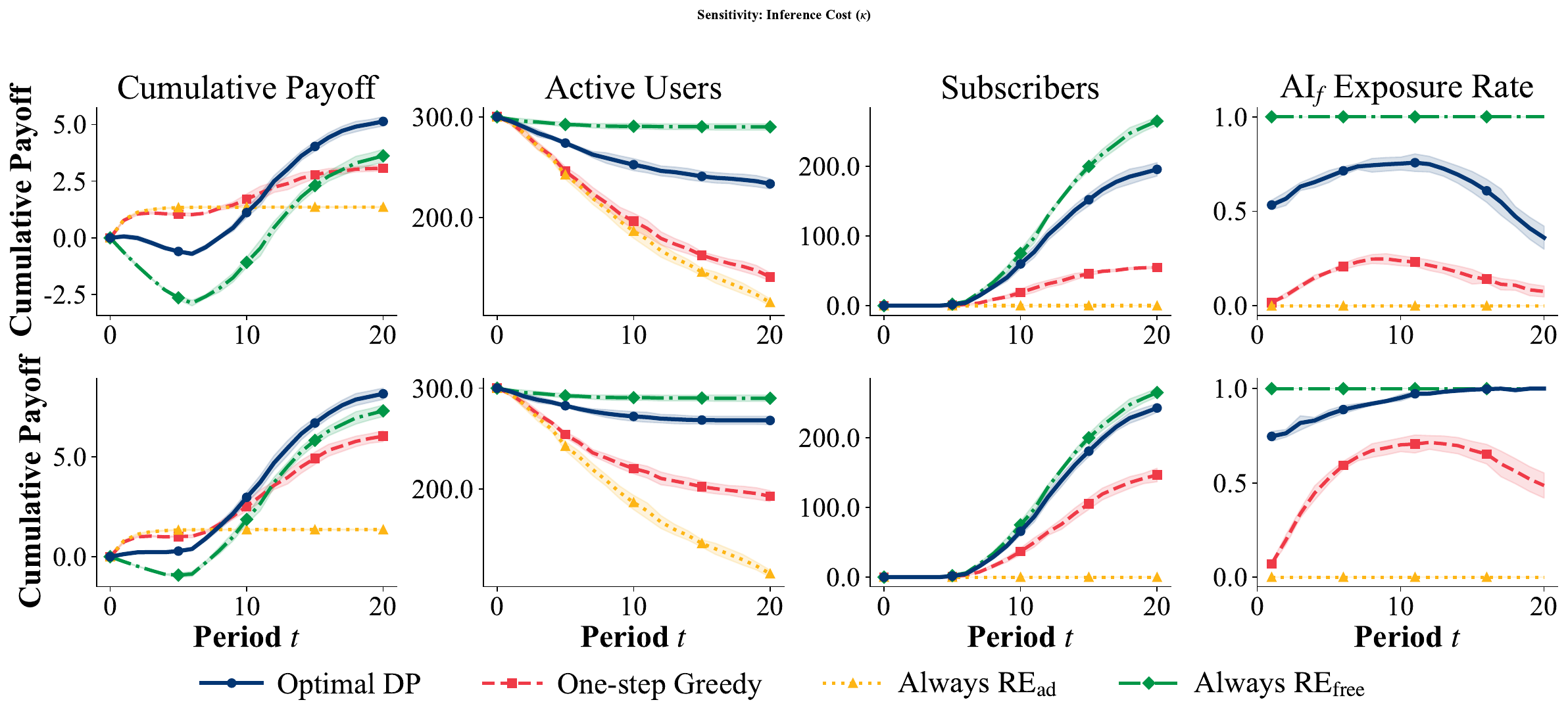}
\caption{Sensitivity to inference cost ($\kappa$): all four metrics under high-$\kappa$ (top) and low-$\kappa$ (bottom) conditions.
Under high inference cost, the DP curtails \AIf\ exposure, yet still outperforms by timing its \AIf\ investments.
Note the negative early payoffs under high $\kappa$ for Always \AIf, reflecting the cost of building the subscription base.}
\label{fig:detail_kappa}
\end{figure}

\begin{figure}[h]
\centering
\includegraphics[width=\linewidth]{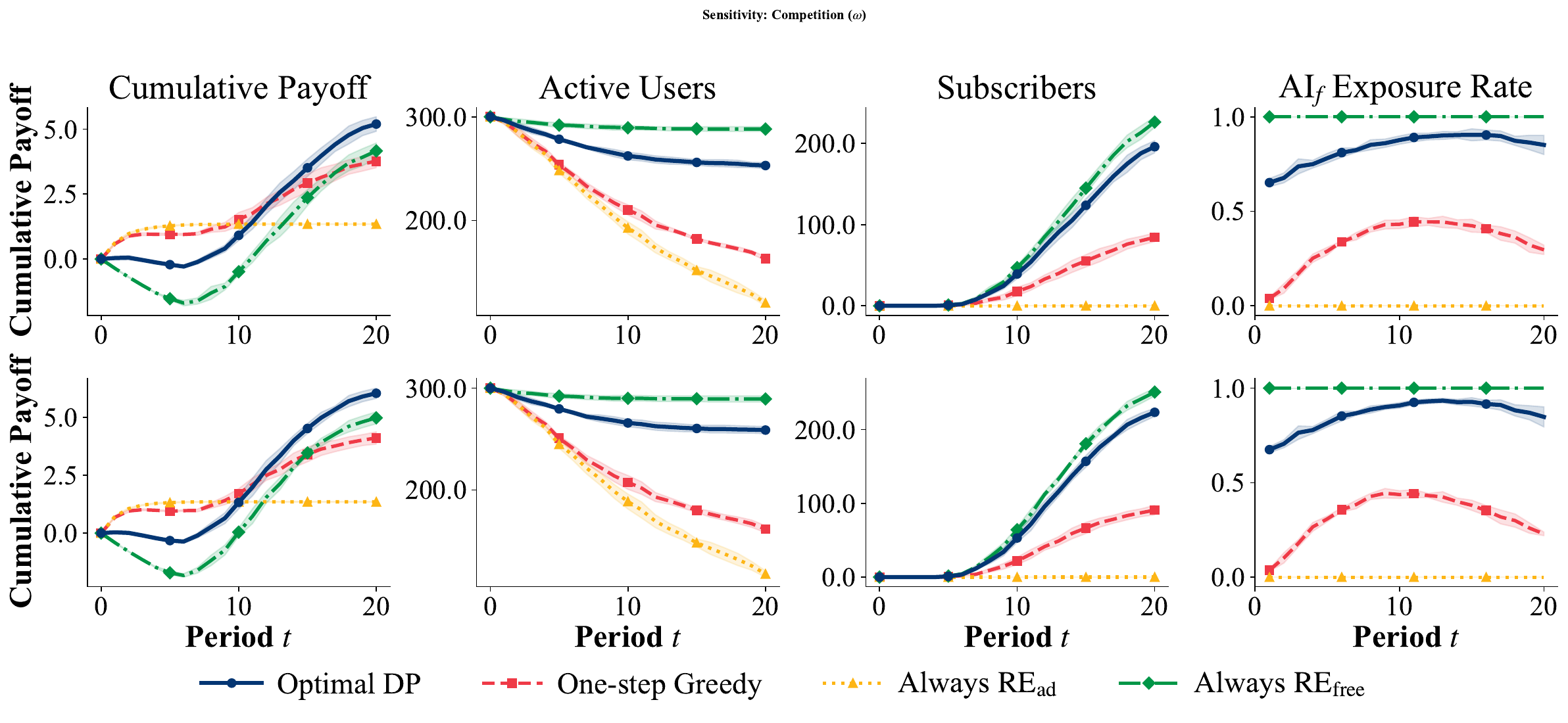}
\caption{Sensitivity to outside competition ($\omega$): all four metrics under high-$\omega$ (top) and low-$\omega$ (bottom) conditions.
Stronger competition compresses margins but preserves the DP's structural advantage, consistent with \Cref{prop:competition_type_shift}.}
\label{fig:detail_omega}
\end{figure}

\begin{figure}[h]
\centering
\includegraphics[width=0.85\linewidth]{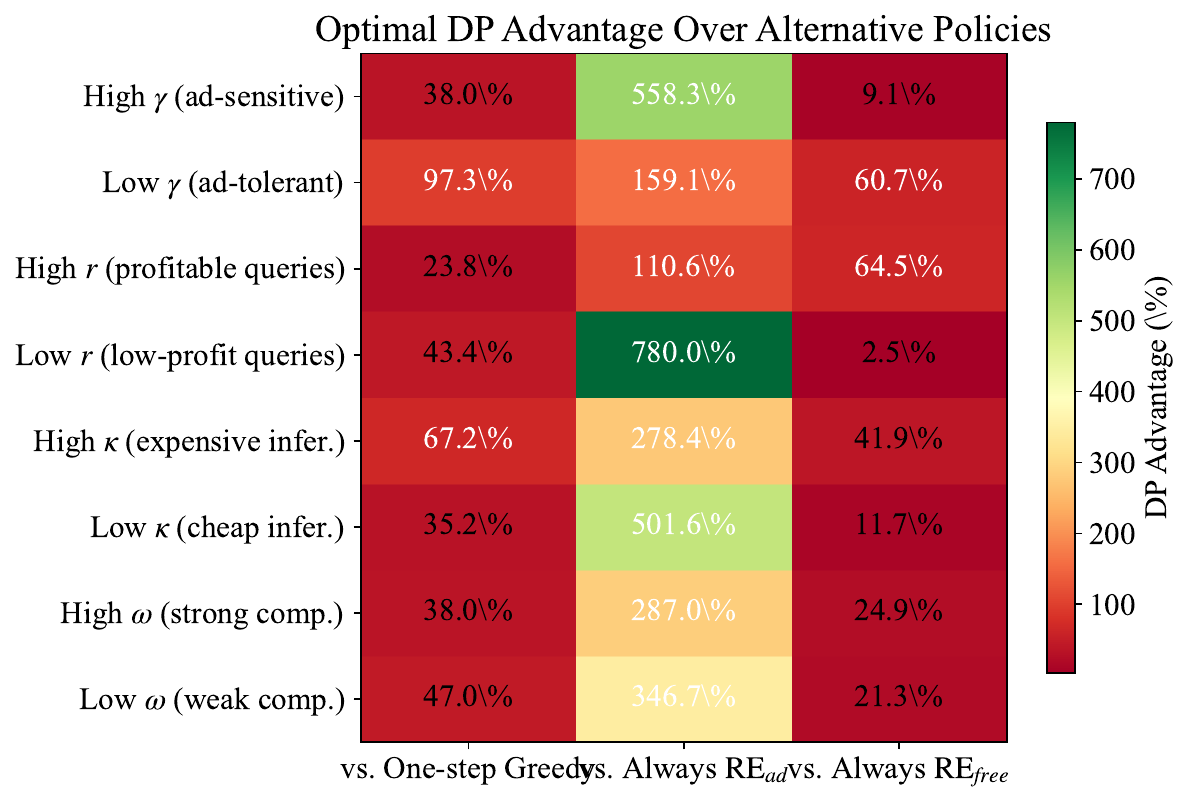}
\caption{Heatmap of the Optimal DP's percentage advantage over each alternative policy across all eight market conditions.
The advantage over Always \AIa\ is dominant across all conditions.}
\label{fig:dp_advantage_heatmap}
\end{figure}





\end{document}